\documentclass[10pt,final]{IEEEtran}
\usepackage{amsfonts}
\usepackage{mathrsfs}
\usepackage{amssymb}
\usepackage{graphicx}
\usepackage{psfrag}
\usepackage{amsmath}
\usepackage{array}
\usepackage{cases}
\usepackage{cite}
\usepackage{tabularx}
\usepackage{url}
\usepackage[level=0]{wgroup_message}

\newcommand{\paperTitle}{The Feasibility Conditions for Interference Alignment in MIMO Networks}
\begin{document}
\author{Liangzhong~Ruan,~\IEEEmembership{Student~Member,~IEEE,
}Vincent~K.N.~Lau, \IEEEmembership{Fellow,~IEEE,}
\\and Moe~Z.~Win, \IEEEmembership{Fellow,~IEEE}
\thanks{This work is funded by Hong Kong Research Grants Council RGC 614910.}
\thanks{L.\ Ruan is with the Electronic and Computer Engineering Department, Hong Kong University of Science and Technology. The author is also a visiting student at Massachusetts Institute of Technology (e-mail: \texttt{lruan@mit.edu}).}
\thanks{V.\ K.\ N.\  Lau is with the Electronic and Computer Engineering Department, Hong Kong University of Science and Technology (e-mail: \texttt{eeknlau@ust.hk}).}
\thanks{M.\ Z.\ Win is with the Laboratory for Information and Decision Systems (LIDS), Massachusetts Institute of Technology (e-mail: \texttt{moewin@mit.edu}).}}

\title{\paperTitle}

\newtheorem{Thm}{Theorem}[section]
\newtheorem{Thm2}{Theorem}[subsection]
\newtheorem{Lem}{Lemma}[section]
\newtheorem{Lem2}{Lemma}[subsection]
\newtheorem{Asm}{Assumption}[section]
\newtheorem{Def}{Definition}
\newtheorem{Remark}{Remark}[section]
\newtheorem{Prob}{Problem}
\newtheorem{Prop}{Proposition}[section]
\newtheorem{Prop2}{Proposition}[subsection]
\newtheorem{Alg}{Algorithm}
\newtheorem{Sta}{Stage}
\newtheorem{Fact}{Fact}
\newtheorem{Example}{Example}
\newtheorem{Cor}{Corollary}[section]
\definecolor{mblue}{rgb}{0,0,0}

\maketitle

\renewcommand{\IEEEQED}{\IEEEQEDopen}
\begin{abstract}
Interference alignment (IA) has attracted great attention in the last few years for its breakthrough performance in interference networks. However, despite the numerous works dedicated to IA, the feasibility conditions of IA remains unclear for most network topologies. The IA feasibility analysis is challenging as the IA constraints are sets of high-degree polynomials, for which no systematic tool to analyze the solvability conditions exists.
 In this work, by developing a new mathematical framework that maps the solvability of sets of polynomial equations to the \emph{linear} independence of their first-order terms, we propose a sufficient condition that applies to MIMO interference networks with general configurations. We have further proved that this sufficient condition matches with the necessary conditions under a wide range of configurations. These results further consolidate the theoretical basis of IA.
\end{abstract}

\section{Introduction}\label{sec:intro}
\mynote{The breakthrough performance of IA.}

Interference has been a fundamental performance bottleneck in wireless communication. Conventional schemes either treat interference as noise or use channel orthogonalization to avoid interference. However, these schemes are non-capacity achieving in general. Interference alignment (IA), first proposed in \cite{firstIA2}, significantly improves the performance of interference networks by aligning the aggregated interference from multiple sources into a lower dimensional subspace. For instance, in a system with $K$ transmitter-receiver (Tx-Rx) pairs and $N$ antennas at each node, the IA achieves a total throughput which scales as $\mathcal{O}\left(\frac{KN}{2}\log(\mbox{SNR})\right)$ \cite{JafarK1}. This scaling law is optimal and well dominates that of conventional orthogonalization schemes, i.e. $\mathcal{O}\left({N}\log(\mbox{SNR})\right)$. The IA solution in \cite{JafarK1} is also applied to other topologies such as the MIMO-X channels \cite{JafarX2} and MIMO relay channels \cite{Jnl:IArelay} and achieves the optimal throughput scaling law. As such, there is a surge in the research interests of IA.

\mynote{People are more interested in IA with finite dimension due}
\mynote{to practicality concerns.}

To achieve the optimal scaling law of throughput, the IA solution in \cite{JafarK1} requires
$\mathcal{O}((KN)^{2K^2N^2})$ dimensions of signal space, which is realized by time or frequency domain symbol extension. Such symbol extension approach is difficult to implement in practice due to the huge dimensions of the signal space involved.
To overcome this problem, IA designs with signal space dimension limited by the number of antennas, are proposed in \cite{Conf:MIMOC1,Conf:MIMOC2,MIMOCell2,Conf:JafarD1,Conf:IA_Heath} for practical MIMO systems. In the IA designs proposed in \cite{Conf:MIMOC1,Conf:MIMOC2,MIMOCell2}, closed-form solutions are obtained for few specific and simple configurations. For instance, in \cite{Conf:MIMOC1}, all Rx have 2 antennas. In \cite{Conf:MIMOC2}, all nodes have {$(K+1)$} antennas. And in \cite{MIMOCell2}, there are only 2 Rxs in the network. Moreover, in all the works mentioned above, each Tx only has one independent data stream. Iterative IA solutions based on alternating optimization
 are proposed for MIMO interference networks with general configurations in \cite{Conf:JafarD1,Conf:IA_Heath}. However, these approaches may not converge to the global optimal solution.

\mynote{With finite signal dimension, existing works on IA feasibility}
\mynote{are quite limited.}

 When the signal space dimension is limited, the IA is not always feasible. Therefore, the characterization of the feasibility conditions under limited signal space dimension is the primary issue to address. In general, the feasibility of the IA problem is associated with the solvability of a set of polynomial equations, which is the focus of {\em algebraic geometry} \cite{Bok:AG_Harris,Bok:AG_Cox2}. There are very few works that studied the feasibility condition of IA problems using {\em algebraic geometry} \cite{JafarDf,Misc:IAfeasible_Tse,Conf:IAfeasible_Luo,Jul:IAfeasible_Luo}.
 In \cite{JafarDf}, the authors studied the feasibility condition of IA problem in single stream MIMO interference networks using {\em Bernstein's Theorem} in algebraic geometry \cite[Thm. 5.4, Ch. 7]{Bok:AG_Cox2}. This work has been extended to the multiple stream case by two parallel works \cite{Misc:IAfeasible_Tse}, and {[14,15]}\mysubsubnote{numbers need to change manually}, respectively. The first approach in \cite{Misc:IAfeasible_Tse} established some necessary conditions for the IA feasibility condition for general network topology by analyzing the dimension of the {\em algebraic varieties} \cite{Bok:AG_Harris}. The authors further showed that these conditions are also sufficient when the number of antennas and data streams at every node are identical. The second approach in \cite{Conf:IAfeasible_Luo,Jul:IAfeasible_Luo} established a similar necessary conditions for the IA feasibility problem based on {\em algebraic independence} between the IA constraints. The authors further proved that these conditions are also sufficient when the number of data stream at every node is the same and the number of antennas at every node is divisible by the number of data streams. In summary, the aforementioned works have proposed some necessary conditions for MIMO interference networks with general configuration, but the proposed sufficient conditions are limited to specific configurations.

\mynote{Our contribution: New tool and its impact.}

In this paper, we develop new tools in algebraic geometry which allows us to address the IA feasibility issue in the general configuration. The newly developed tool maps the solvability of a set of general polynomial equations to the {\em linear independence} of their first order terms. Based on this new tool, we can extend our understanding on the IA feasibility conditions in the following aspects:
\begin{itemize}
\item[\bf A.] Further tighten the IA feasibility conditions from the {\em necessary} side;
\item[\bf B.] Propose and prove a {\em sufficient} condition of IA feasibility which applies to MIMO interference networks with {\em general} configurations;
{\item[\bf C.] Prove that scaling the number of antennas and data streams of a network simultaneously preserves IA feasibility;}
\item[\bf D.] Determine the necessary and sufficient conditions of IA feasibility in a wider range of network configurations comparing with the results given in \cite{Misc:IAfeasible_Tse,Conf:IAfeasible_Luo,Jul:IAfeasible_Luo}.
\end{itemize}
\mynote{Paper outline}

Organization: Section~\ref{sec:model} presents the system model and define the IA feasibility problem. Section~\ref{sec:mainresult} pairs the analytical results of this paper
to their contributions.  Section~\ref{sec:proof} provides the  proofs of the results based on a new mathematical framework. Section~\ref{sec:conclude} gives the conclusion.

\mynote{Notation}

Notations: $a$, $\mathbf{a}$, $\mathbf{A}$, and $\mathcal{A}$ represent scalar, vector, matrix, set/space, respectively.  $\mathbb{N}$, $\mathbb{Z}$ and $\mathbb{C}$ denote the set of natural numbers, integers and complex numbers, respectively. The operators $(\cdot)^T$, $(\cdot)^H$, {$\det(\cdot)$}, $\mbox{rank}(\cdot)$, and $\mathcal{N}(\cdot)$ denote transpose, Hermitian transpose, {determinate}, rank, and null space of a matrix. And the operators $\langle\cdot \rangle$, $\mathcal{V}(\cdot)$ denote the ideal, and the vanishing set \cite{Bok:AG_Cox} of a set of polynomials. { For a field $\mathcal{K}$, $\mathcal{K}[x_1,...x_j]$ represents the field of rational functions in variables $x_1,...,x_j$ with coefficients drawn from $\mathcal{K}$.} $\mbox{size}(\cdot)$ represents the size of a vector and {$|\cdot|$ represents the cardinality of a set.} $\mbox{I}(\cdot)$ is the indicator function. $\dim(\cdot)$ denotes the dimension of a space. $\mbox{span}(\mathbf{A})$ and $\mbox{span}(\{\mathbf{a}\})$ denote the linear space spanned by the column vectors of $\mathbf{A}$ and the vectors in set $\{\mathbf{a}\}$, respectively.  { $\mbox{gcd}(n,m)$ denotes the greatest common divisor of $n$ and $m$, $n|m$ denotes that $n$ divides $m$, and $\mbox{mod}(n,m)$ denotes $n$ modulo $m$, $n,m\in\mathbb{Z}$. $\mbox{diag}^{n}(\mathbf{A},...,\mathbf{X})$ represents a block diagonal matrix with submatrixes $\mathbf{A},...,\mathbf{X}$ on its  $n$-th diagonal. For instance, $\mbox{diag}^{-\!1}([2,1],[1,2])=\Bigg[\!\!{\scriptsize\begin{array}{*{5}{c@{\,}}c}0&0&0&0&0&0\\2&1&0&0&0&0\\0&0&1&2&0&0\end{array}}\!\!\Bigg]$. $\mbox{diag}(\mathbf{A},...,\mathbf{X})=\mbox{diag}^{0}(\mathbf{A},...,\mathbf{X})$, and $\mbox{diag}[n](\mathbf{A})=\mbox{diag}(\underbrace{\mathbf{A},...,\mathbf{A}}_{n\mbox{ \scriptsize times}})$.} The symbol ``$\cong$" denotes the isomorphism relation \cite{Misc:Springer_Plucker}.

\section{System Model and Problem Formulation}
\label{sec:model}
\mynote{Channel Model}

Consider a MIMO interference network consisting of $K$ Tx-Rx pairs, with Tx $k$ sending $d_k$ independent data streams to Rx $k$. Tx $k$ is equipped with $M_k$ antennas and Rx $k$ has $N_k$ antennas. The received signal $\mathbf{y}_{k}\in \mathbb{C}^{d_k}$ at Rx $k$ is given by:
\begin{eqnarray}
\mathbf{y}_{k}=\mathbf{U}^H_{k}\left(\mathbf{H}_{kk}
\mathbf{V}_{k}\mathbf{x}_{k} + \sum_{j=1,\neq k}^{K}
\mathbf{H}_{kj} \mathbf{V}_{j}\mathbf{x}_{j}+\mathbf{z}
\right)\label{eqn:signal_1}
\end{eqnarray}
where $\mathbf{H}_{kj}\in \mathbb{C}^{N_k\times M_j}$ is the channel state matrix from Tx $j$ to Rx $k$, whose entries are independent random variables drawn from continuous distributions. $\mathbf{x}_k \in \mathbb{C}^{d_k}$  is the encoded
information symbol for Rx $k$, $\mathbf{U}_{k}\in
\mathbb{C}^{N_k \times d_k}$  is the decorrelator of Rx
$k$, and $\mathbf{V}_{j}\in \mathbb{C}^{M_j \times
d_j}$  is the transmit precoding matrix at Tx
$j$. $\mathbf{z}\in \mathbb{C}^{N_k\times 1}$ is the
white Gaussian noise with zero mean and unit variance.

\mynote{Problem Formulation}

Following the previous works on IA for $K$-pairs MIMO interference networks \cite{JafarDf_c,JafarDf,Misc:IAfeasible_Tse,Conf:IAfeasible_Luo,Jul:IAfeasible_Luo}, in this work, we focus on the feasibility issue of the following problem:
\begin{Prob}[IA on MIMO Interference Networks]\label{pro:pIA_c}
For a MIMO interference network with configuration $\chi=\{(M_1,N_1,d_1),(M_2,N_2,d_2),...,(M_K,N_K,d_K)\}$,
design transceivers $\{\mathbf{U}_k\in\mathbb{C}^{N_k\times d_k}, \mathbf{V}_j\in\mathbb{C}^{M_j\times d_j}\}$, $k,j\in\{1,...,K\}$ that satisfy the following constraints:
\begin{eqnarray}
\mbox{rank}\left(\mathbf{U}^H_k\mathbf{H}_{kk}\mathbf{V}_k\right) &\!\!\!=\!\!\!& d_k\label{eqn:drank_c}, \qquad\forall k,\\
\mathbf{U}^H_k\mathbf{H}_{kj}\mathbf{V}_j &\!\!\!=\!\!\!& \mathbf{0},\qquad\hspace{1.5mm}\forall k\neq j.\label{eqn:czero_c}
\end{eqnarray}
\end{Prob}

\section{Feasibility Conditions}
\label{sec:pIA}
\mynote{Highlight the outline of this section,}
In this section, we will first list the main results and pair them with the contributions. Then we prove these results in the second subsection. Readers can refer to \cite{Tech:Ruan} for a summary of the main theoretical approaches prior to this work, and a  brief introduction to the concept of algebraic independence.
\subsection{Main Results}\label{sec:mainresult}
\subsubsection{Theorems Applicable to General Configurations}
The following two theorems summarize the main result on the necessary side and the sufficient side, respectively.
\begin{Thm}[Necessary Conditions of IA Feasibility]\label{thm:feasible}
If Problem~\ref{pro:pIA_c} has solution, then the network configuration $\chi$  must satisfy the following inequalities:
\begin{eqnarray}
&&\!\!\!\!\min\{M_j,N_j\} \ge d_j,\,{\forall j\in\{1,...,K\}} \label{eqn:f1}
\\&&\!\!\!\!\max\{\sum_{j:(\cdot,j)\in\atop \mathcal{J}_{\mathrm{sub}}}\!\!M_j,\sum_{k:(k,\cdot)\in\atop \mathcal{J}_{\mathrm{sub}}}\!\!N_k\} \ge \!\!\sum_{j:\;(\cdot,j) {\scriptsize\mbox{ or }} (j,\cdot)\in\atop \mathcal{J}_{\mathrm{sub}}}\!\!d_j, \label{eqn:f2}
\\&&\!\!\!\!\!\!\sum_{j:(\cdot,j)\in\atop \mathcal{J}_{\mathrm{sub}}}\!\!\!d_j(M_j-d_j)
+\!\!\!\!\!\sum_{k:(k,\cdot)\in\atop \mathcal{J}_{\mathrm{sub}}}\!\!\!d_k(N_k-d_k)\ge\!\!\sum_{(k,j)\in\atop\mathcal{J}_{\mathrm{sub}}}\!
d_j d_k,\label{eqn:f3}
\end{eqnarray}
$\forall \mathcal{J}_{\mathrm{sub}}\subseteq \mathcal{J}$, where $\mathcal{J}=\{(k,j):\; k,j\in\{1,...,K\},k\neq j\}$, $(\cdot,j)$ (or $(k,\cdot)$) $\in\mathcal{J}_{\mathrm{sub}}$ denote that there exists $k$ (or $j$) $\in\{1,...,K\}$ such that $(k,j)$ $\in \mathcal{J}_{\mathrm{sub}}$. {~\hfill\IEEEQED}
\end{Thm}

\begin{Remark}[Tighter Necessary Conditions] \eqref{eqn:f2} is the newly proposed necessary condition. If the cardinality of set $\mathcal{J}_{\mathrm{sub}}$ is restricted to be $1$, we have that \eqref{eqn:f2} is reduced to
\begin{eqnarray}
\max\{M_j,N_k\}\ge d_k+d_j,\;\forall j\neq k, \label{eqn:f2_old}
\end{eqnarray}
which is one of the necessary inequalities given in the prior works \cite{Misc:IAfeasible_Tse,Conf:IAfeasible_Luo,Jul:IAfeasible_Luo}. Note that the necessary conditions given in Thm.~\ref{thm:feasible} are strictly tighter than those given in \cite{Misc:IAfeasible_Tse,Conf:IAfeasible_Luo,Jul:IAfeasible_Luo}.~\hfill\IEEEQED
\end{Remark}

{\begin{Thm}[Sufficient Condition of IA Feasibility]\label{thm:feasible_s}
If the matrix described in Fig.~\ref{fig_Hall} (denote this matrix as $\mathbf{H}_{\mathrm{all}}$) is full row rank, Problem~\ref{pro:pIA_c} has solutions almost surely.
\begin{figure*}[t] \centering
\includegraphics[scale=0.75]{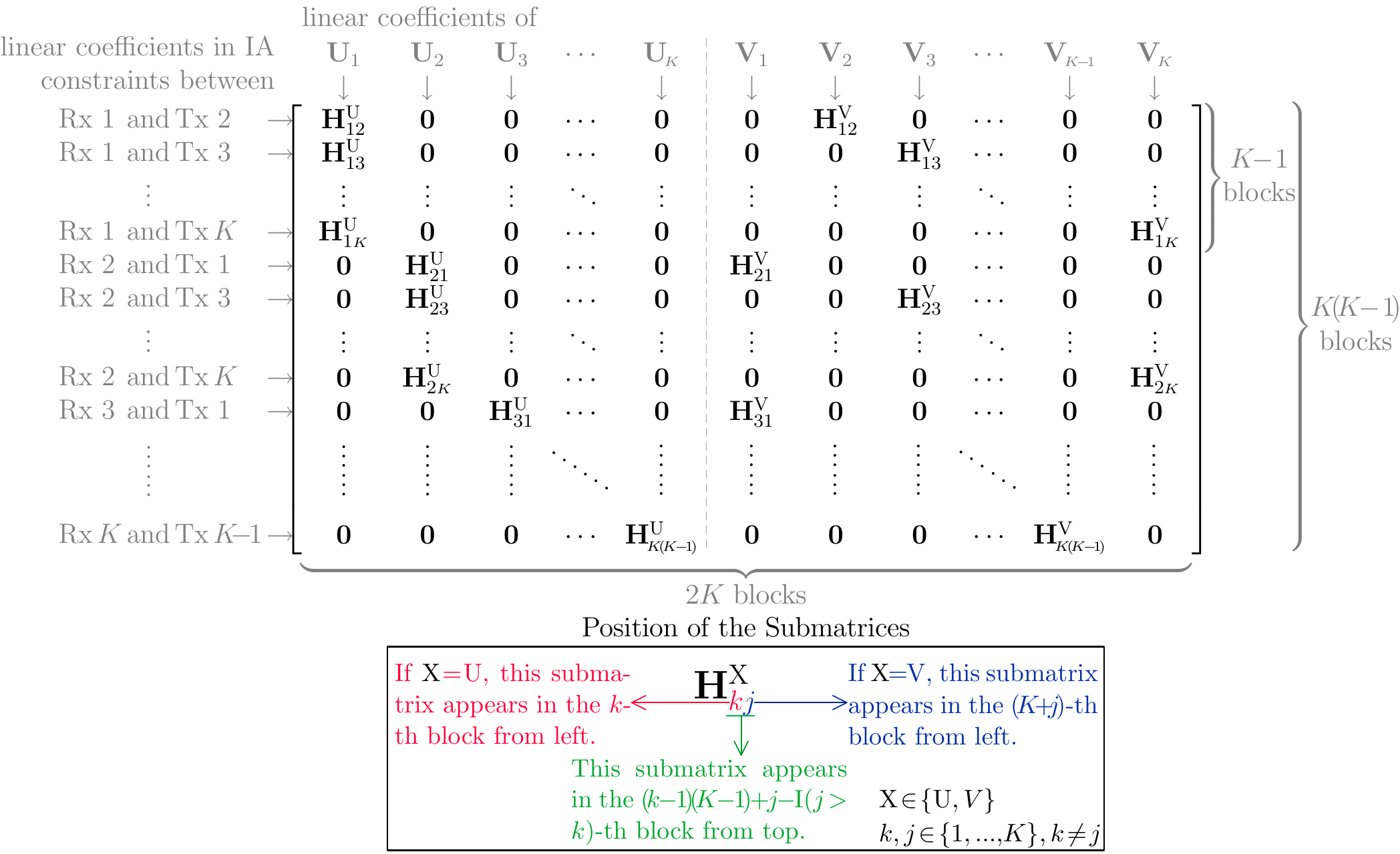}
\caption {The matrix scattered by the coefficient vectors of the linear terms in the polynomial form of IA constraints, i.e. \eqref{eqn:czero_poly2}.}
\label{fig_Hall}
\end{figure*}

The submatrices $\mathbf{H}^{\mathrm{U}}_{kj}\in\mathbb{C}^{(d_kd_j)\!\times\!(d_k\!(M_k\!-\!d_k))}$, $\mathbf{H}^{\mathrm{V}}_{kj}\in\mathbb{C}^{(d_kd_j)\!\times\!(d_j\!(N_j\!-\!d_j))}$ in Fig.~\ref{fig_Hall} are defined by:
\begin{eqnarray}
\mathbf{H}^{\mathrm{U}}_{kj}\!\!\!\!\!\!&=&\!\!\!\!\!\!
\mbox{diag}[d_k]\!\!\!\left(\!\!\!\!\begin{array}{*{3}{c@{\,}}c}
h_{kj}\!(\!d_k\!\!+\!\!1,\!1\!),&h_{kj}\!(\!d_k\!\!+\!\!2,\!1\!),
&\!\cdots\!,&h_{kj}\!(\!N_k,\!1\!)\\
h_{kj}\!(\!d_k\!\!+\!\!1,\!2\!),&h_{kj}\!(\!d_k\!\!+\!\!2,\!2\!),
&\!\cdots\!,&h_{kj}\!(\!N_k,\!2\!)\\
\vdots&\vdots&\ddots&\vdots\\
h_{kj}\!(\!d_k\!\!+\!\!1,\!d_j\!),&h_{kj}\!(\!d_k\!\!+\!\!2,\!d_j\!),
&\!\cdots\!,&h_{kj}\!(\!N_k,\!d_j\!)
\end{array}\!\!\!\!\right)\label{eqn:hu}\\
\mathbf{H}^{\mathrm{V}}_{kj}\!\!\!\!\!\!&=&\!\!\!\!\!\!\!\!\left[\!\!\!\!
\begin{array}{*{4}{c@{\,}}c}
\mbox{diag}[d_j]\!\big(&\!\!
h_{kj}\!(\!1,\!d_j\!\!+\!\!1\!),&h_{kj}\!(\!1,\!d_j\!\!+\!\!2\!),&\!\cdots\!,&h_{kj}\!(\!1,\!M_j\!)\big)\\
\mbox{diag}[d_j]\!\big(&\!\!
h_{kj}\!(\!2,\!d_j\!\!+\!\!1\!),&h_{kj}\!(\!2,\!d_j\!\!+\!\!2\!),&\!\cdots\!,&h_{kj}\!(\!2,\!M_j\!)\big)\\
\multicolumn{5}{c}{\cdots\cdots}\\
\mbox{diag}[d_j]\!\big(&\!\!
h_{kj}\!(\!d_k,\!d_j\!\!+\!\!1\!),&h_{kj}\!(\!d_k,\!d_j\!\!+\!\!2\!),&\!\cdots\!,&h_{kj}\!(\!d_k,\!M_j\!)\big)
\end{array}
\!\!\!\!\!\right]\label{eqn:hv}
\end{eqnarray}
where $h_{kj}(p,q)$ denotes the element in the $p$-th row and $q$-th column of $\mathbf{H}_{kj}$, $k\neq j, k,j\in\{1,...,K\}$.{~\hfill\IEEEQED}
\end{Thm}}

\begin{Remark}[Interpretation of the Sufficient Condition] The row vectors of $\mathbf{H}_{\mathrm{all}}$ are the coefficients of the linear terms of polynomials in the IA constraint \eqref{eqn:czero_c}. Please refer to \eqref{eqn:czero_poly2}, \eqref{eqn:v} for details. Hence, Thm.~\ref{thm:feasible_s} claims that the linear independence of these coefficient vectors is sufficient for the IA problem to be feasible. This fact is a direct consequence of the mathematical tool we developed in algebraic geometry, i.e. Lem.~\ref{lem:independent1}--\ref{lem:nonempty}. Please refer to Fig.~\ref{fig_flow} for an intuitive illustration of this mathematical tool.~\hfill\IEEEQED
\end{Remark}

\begin{figure*}[t] \centering
\includegraphics[scale=0.65]{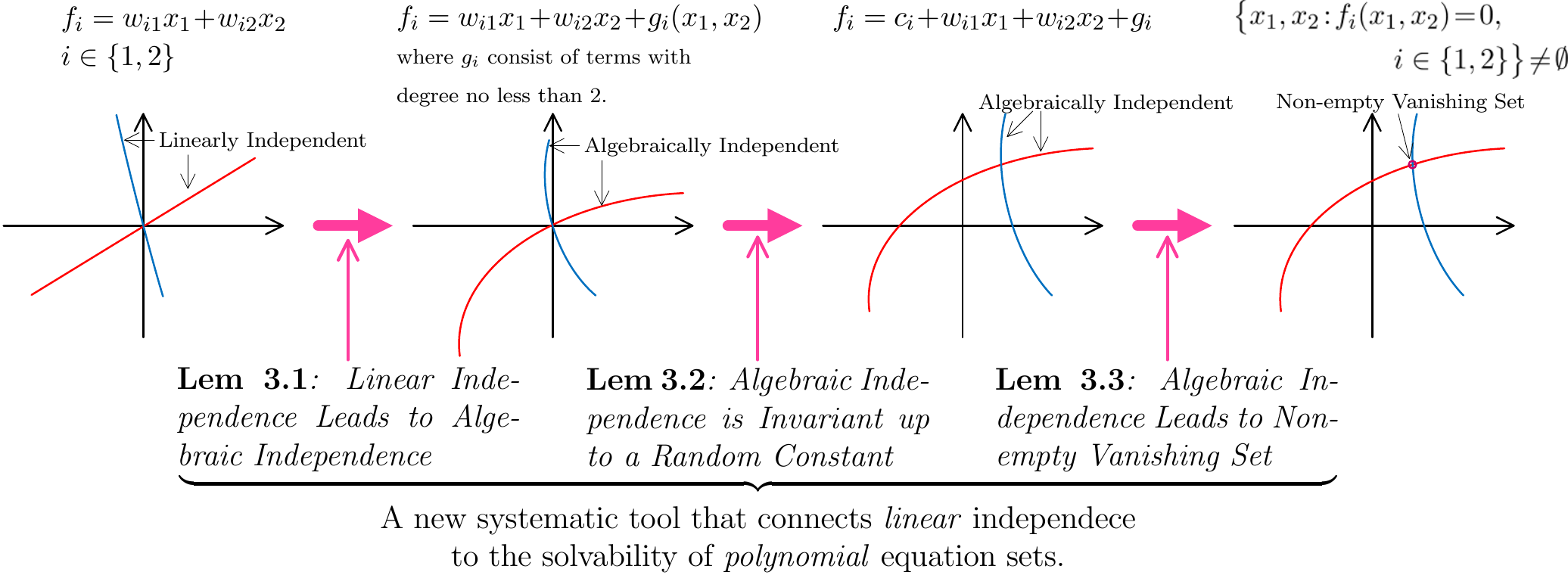}
\caption {Illustration of the new systematic tool that links \emph{linear} independence to the solvability of \emph{polynomial} equation sets.}
\label{fig_flow}
\end{figure*}

\begin{Remark}[Contributions of Thm.~\ref{thm:feasible_s}] In literature,  sufficient conditions of IA feasibility are limited to special network configurations. Thm.~\ref{thm:feasible_s} proposes a sufficient condition which applies to MIMO interference networks with general configuration.~\hfill\IEEEQED
\end{Remark}

{
Following Thm.~\ref{thm:feasible_s}, we have two corollaries that depict the relation between network configuration $\chi$ and IA feasibility.

\begin{Cor}[Configuration Dominates IA Feasibility]\label{cor:confdomin}
Under a network configuration $\chi$, $\mathbf{H}_{\mathrm{all}}$ is either always rank deficient or full row rank almost surely. Hence, if $\mathbf{H}_{\mathrm{all}}$ is full row rank under one realization of channel state $\{\mathbf{H}_{kj}\}$, Problem~\ref{pro:pIA_c} has solution almost surely in this network.{~\hfill\IEEEQED}
\end{Cor}

Following Cor.~\ref{cor:confdomin}, we define the notation of ``IA feasible network", if Problem~\ref{pro:pIA_c} has solution almost surely in this network.

\begin{Remark}[Numerical and Analytical Contributions] Cor.~\ref{cor:confdomin} highlights the fact that the network configuration $\chi$, rather than the specific channel state $\{\mathbf{H}_{kj}\}$ dominates the IA feasibility. This phenomenon is useful in both numerical test and theoretical analysis: In practice, to test the IA feasibility of a specific network, we only need to randomly generate \emph{one} channel state and check if $\mathbf{H}_{\mathrm{all}}$ is full rank. Similarly, to prove that a certain category of network is IA feasible, we can try to construct some specific channel state that makes $\mathbf{H}_{\mathrm{all}}$ full row rank for all the networks in this category. In fact, we will exploit this property in the proof of Cor.~\ref{cor:sym}.~\hfill \IEEEQED
\end{Remark}
\mynote{The new result that we get this time.}
\begin{Cor}[Scalability of IA Feasibility]\label{cor:scaleable} If a network with configuration $\chi=\{(M_1,N_1,d_1),...,(M_K,N_K,d_K)\}$ is IA feasible, then scaling it by a factor, i.e. $\chi^c=\{(cM_1,cN_1,cd_1),...,(cM_K,cN_K,cd_K)\}$, $c\in\mathbb{N}$ preserves its IA feasibility.{~\hfill\IEEEQED}
\end{Cor}
}

\subsubsection{Corollaries Applicable to Special Configurations}
In the following analysis, we show that the necessary conditions in Thm.~\ref{thm:feasible} and the sufficient condition in Thm.~\ref{thm:feasible_s} {match} in some network configurations. The conditions align in wider range of configurations than the existing results in \cite{JafarDf,Misc:IAfeasible_Tse,Conf:IAfeasible_Luo,Jul:IAfeasible_Luo}.

\begin{Cor}[Symmetric Case]\label{cor:sym}When the network configuration $\chi$ is symmetric, i.e. $d_k\!=\!d$, $M_k\!=\!M$, $N_k\!=\!N$, $\forall k\!\in\!\{1,...,\!K\}$, and $\min\{M,N\}\ge 2d$, Problem~\ref{pro:pIA_c} has solution almost surely if and only if inequality \eqref{eqn:f3_sym} is true, where
{\begin{eqnarray}M+N-(K+1)d\ge0. \label{eqn:f3_sym}
\end{eqnarray}\vspace{-10pt}}

~\hfill\IEEEQED
\end{Cor}

\begin{Remark}[Backward Compatible to \cite{Misc:IAfeasible_Tse}] If we further assume that $M=N$ and $K\ge 3$, the feasibility conditions in Cor.~\ref{cor:sym} is reduced to $2N-(K+1)d\ge0$, which is consistent with the IA feasibility conditions given in \cite{Misc:IAfeasible_Tse}.
\end{Remark}

\begin{Cor}[``Divisible" Case]\label{cor:div}When the network configuration $\chi$ satisfies 1) $d_k=d$, $\forall k$, and 2) $d|N_k$, $\forall k$ or $d|M_k$, $\forall k$, Problem~\ref{pro:pIA_c} has solution almost surely if and only if inequality \eqref{eqn:f3_divisor} is satisfied, where\vspace{-5pt}
\begin{eqnarray}
\!\!\!\sum_{j:(\cdot,j)\in\mathcal{J}_{\mathrm{sub}}}\!\!\!\!\!\!\!\!(M_j-d)
+\!\!\!\!\!\!\!\sum_{k:(k,\cdot)\in\mathcal{J}_{\mathrm{sub}}}\!\!\!\!\!\!\!\!(N_k-d)\ge d|\mathcal{J}_{\mathrm{sub}}|,\;\forall\mathcal{J}_{\mathrm{sub}}\subseteq \mathcal{J}.\label{eqn:f3_divisor}
\end{eqnarray}\vspace{-10pt}

~\hfill\IEEEQED
\end{Cor}

\begin{Remark}[Backward Compatible to  \cite{JafarDf}, \cite{Conf:IAfeasible_Luo}, and \cite{Jul:IAfeasible_Luo}] If we further assume that $d=1$, then $d|M_k$, $d|N_k$ for all $k\in\{1,...,K\}$. In this case, Cor.~\ref{cor:div} corresponds to that in \cite{JafarDf}. Similarly, if we require both $N_k$ and $M_k$ are divisible by $d$, Cor.~\ref{cor:div} is reduced to the feasibility conditions given by \cite{Conf:IAfeasible_Luo,Jul:IAfeasible_Luo}.~\hfill\IEEEQED
\end{Remark}

\subsection{Proof of the Feasibility Conditions}
\label{sec:proof}
\mynote{Prove Thm.~\ref{thm:feasible} by analyzing the dimension of spaces.}
\subsubsection{Proof of Theorem~\ref{thm:feasible}}
Note that the necessity of \eqref{eqn:f1} and \eqref{eqn:f3} is proved in \cite{Misc:IAfeasible_Tse,Conf:IAfeasible_Luo,Jul:IAfeasible_Luo}{.}  We need to prove the necessity of \eqref{eqn:f2}.

Suppose Problem~\ref{pro:pIA_c} is feasible. {Without loss of generality, assume for a certain set $\mathcal{J}^\dag_{\mathrm{sub}}\subseteq \mathcal{J}$,
\begin{eqnarray}\sum_{k:(k,\cdot)\in \mathcal{J}^\dag_{\mathrm{sub}}}\!\!\!\!N_k=\max\bigg\{\sum_{j:(\cdot,j)\in \mathcal{J}^\dag_{\mathrm{sub}}}\!\!\!\!M_j,\sum_{k:(k,\cdot)\in\mathcal{J}^\dag_{\mathrm{sub}}}
\!\!\!\!N_k\bigg\}.\label{eqn:antenna}
\end{eqnarray}
Then for $\mathcal{J}^\dag_{\mathrm{sub}}$,} \eqref{eqn:f2} can be rewritten as:
\begin{eqnarray}\sum_{k:(k,\cdot)\in \mathcal{J}^\dag_{\mathrm{sub}}}\!\!\!\!N_k \ge\!\!\!\! \sum_{j:\;(\cdot,j) {\scriptsize\mbox{ or }} (j,\cdot)\in \mathcal{J}^\dag_{\mathrm{sub}}}\!\!\!\!\!\!\!d_j\label{eqn:ftarget}\end{eqnarray}

We will prove that if Problem~\ref{pro:pIA_c} has a solution, \eqref{eqn:ftarget} must be true. Denote $\mathcal{T}$ as the set of the indices which appears in $\mathcal{J}^\dag_{\mathrm{sub}}$ as Tx index but not Rx index, i.e. $\mathcal{T}\triangleq\{j_1,...,j_m\}=\{j: {\exists k\,\mbox{s.t.}}\,(k,j)\in \mathcal{J}^\dag_{\mathrm{sub}} \mbox{ and } (j,k)\not\in\mathcal{J}^\dag_{\mathrm{sub}}\}$, and denote $\mathcal{R}$ as {the set of indices} which appears in $\mathcal{J}^\dag_{\mathrm{sub}}$ as Rx index, i.e. $\mathcal{R}\triangleq\{k_1,...,k_n\}=\{k:{\exists j\,\mbox{s.t.}}\,(k,j)\in \mathcal{J}^\dag_{\mathrm{sub}}\}$. Denote $\{\mathbf{U}^*_k,\mathbf{V}^*_k\}$ as one of the solution. Construct three matrices:
\begin{eqnarray*}
\mathbf{V}^*_{\mathcal{T}}\!\!\!\!\!&=&\!\!\!\!\!\!\left[\!\!\!\begin{array}{c@{,}c@{,}c@{,}c}
\mathbf{V}^*_{j_1}&\mathbf{0}&\cdots&\mathbf{0}\\
\mathbf{0}&\mathbf{V}^*_{j_2}&\cdots&\mathbf{0}\\
&\multicolumn{2}{c}{\cdots\cdots}&\\
\mathbf{0}&\mathbf{0}&\cdots&\mathbf{V}^*_{j_n}\end{array}\!\!\!\right]\!\!,\,
\mathbf{U}^*_{\mathcal{R}}\!=\!\left[\!\!\!\begin{array}{c@{,}c@{,}c@{,}c}
\mathbf{U}^*_{k_1}&\mathbf{0}&\cdots&\mathbf{0}\\
\mathbf{0}&\mathbf{U}^*_{k_2}&
\cdots&\mathbf{0}\\
&\multicolumn{2}{c}{\cdots\cdots}&\\
\mathbf{0}&\mathbf{0}&\cdots&\mathbf{U}^*_{k_m}\end{array}\!\!\!\right]\!\!,
\\
\mathbf{H}_{\mathcal{J}_{\mathrm{sub}}}\!\!\!\!\!&=&\!\!\!\!\!\!\left[\!\!\!\begin{array}{c@{,}c@{,}c@{,}c}
\mathbf{H}_{k_1j_1}&\mathbf{H}_{k_1j_2}&\cdots&\mathbf{H}_{k_1j_n}\\
\mathbf{H}_{k_2j_1}&\mathbf{H}_{k_2j_2}&\cdots&\mathbf{H}_{k_2j_n}\\
&\multicolumn{2}{c}{\cdots\cdots}&\\
\mathbf{H}_{k_mj_1}&\mathbf{H}_{k_mj_2}&\cdots&\mathbf{H}_{k_mj_n}\end{array}\!\!\!\right]\!\!.
\end{eqnarray*}

Then from \eqref{eqn:drank_c} and \eqref{eqn:czero_c}, we have that:
\begin{eqnarray}
&\!\!\!\!&\!\!\!\!\!\!\!\!\!\!\!\!\dim\left(\mbox{span}(\mathbf{V}^*_{\mathcal{T}})\right)=\!\sum_{j\in\mathcal{T}}\!d_j
,\;\dim\left(\mbox{span}(\mathbf{U}^{*}_{\mathcal{R}})\right)=\!\sum_{k\in\mathcal{R}}\!d_k,\label{eqn:spandim}\\
&\!\!\!\!&\!\!\!\!\!\!\!\!\!\!\!\! \mathbf{U}^{*H}_{\mathcal{R}}\mathbf{H}_{\mathcal{J}^\dag_{\mathrm{sub}}}\mathbf{V}^*_{\mathcal{T}}=\mathbf{0}\label{eqn:null}.
\end{eqnarray}

From \eqref{eqn:antenna}, $\sum_{j\in\mathcal{T}}M_j\le \sum_{j:(\cdot,j)\in \mathcal{J}^\dag_{\mathrm{sub}}}M_j\le\sum_{k\in\mathcal{R}}N_k$, which means in $\mathcal{J}^\dag_{\mathrm{sub}}$, the number of rows is no more than the number of columns. Further note that the elements of $\mathbf{H}_{\mathcal{J}^\dag_{\mathrm{sub}}}$ are independent random variables, we have that $\mathcal{N}(\mathbf{H}_{\mathcal{J}^\dag_{\mathrm{sub}}})=\{\mathbf{0}\}$ almost surely. Therefore
\begin{eqnarray}
\dim\left(\mbox{span}(\mathbf{H}_{\mathcal{J}^\dag_{\mathrm{sub}}}
\mathbf{V}^*_{\mathcal{T}})\right)=\dim\left(\mbox{span}(\mathbf{V}^*_{\mathcal{T}})\right)
=\sum_{j\in\mathcal{T}}d_j\label{eqn:spandim2}
\end{eqnarray}
almost surely. From \eqref{eqn:null}, $\mbox{span}(\mathbf{H}_{\mathcal{J}^\dag_{\mathrm{sub}}}
\mathbf{V}^*_{\mathcal{T}})\perp \mbox{span}(\mathbf{U}^{*}_{\mathcal{R}})$, hence we have:
\begin{eqnarray}
\nonumber \sum_{k\in\mathcal{R}}N_k \!\!\!\!&\ge&\!\!\!\! \dim\left(\mbox{span}(\mathbf{H}_{\mathcal{J}^\dag_{\mathrm{sub}}}
\mathbf{V}^*_{\mathcal{T}})+\mbox{span}(\mathbf{U}^{*}_{\mathcal{R}})\right)
\\\nonumber \!\!\!\!&=&\!\!\!\!\dim\left(\mbox{span}(\mathbf{H}_{\mathcal{J}^\dag_{\mathrm{sub}}}
\mathbf{V}^*_{\mathcal{T}})\right)+\dim\left(\mbox{span}(\mathbf{U}^{*}_{\mathcal{R}})\right)
\\ \!\!\!\!&=&\!\!\!\!\sum_{j\in\mathcal{T}}d_j+\sum_{k\in\mathcal{R}}d_k=\!\!\!\! \sum_{j:\;(\cdot,j) {\scriptsize\mbox{ or }} (j,\cdot)\in \mathcal{J}^\dag_{\mathrm{sub}}}\!\!\!\!\!\!\!d_j\label{eqn:ftarget_2}
\end{eqnarray}

From \eqref{eqn:ftarget_2}, \eqref{eqn:ftarget} is true. This completes the proof.

\mynote{Prove Thm.~\ref{thm:feasible_s} by adopting the new tool developed}
\mynote{from algebraic geometry}
\subsubsection{Proof of Theorem~\ref{thm:feasible_s}}

\mysubnote{Emphasize the importance of the new tool.}

The IA feasibility issue is challenging as there is no systematic tool to address the solvability issue of high-degree polynomial equation sets. In the following analysis, we first elaborate three lemmas. As illustrated in Fig.~\ref{fig_flow}, these lemmas construct a new systematic tool that links \emph{linear} independence to the solvability of \emph{polynomial} equation sets. The newly developed tool is not only the key steps to handle the IA feasibility issue in this work, but also a good candidate of handling the solvability issue of sets of polynomial equations in general.

\mysubnote{Elaborate the new tool.}

\begin{Lem}\emph{(Linear Independence Leads to Algebraic Independence)} \label{lem:independent1} Suppose $\mathcal{K}$ is an algebraically closed field. Consider { $L$ polynomials $f_i\in \mathcal{K}[x_1,x_2,...x_{S}]$, $i\in\{1,...,L\}$} which are given by: $f_i=\sum_{j=1}^{{S}}h_{ij}x_j + g_i$, where $g_i$ are polynomials consisting of terms with degree no less than $2$. If the coefficient vectors $\mathbf{h}_i=[h_{i1},h_{i2},...,h_{i{S}}]$ are linearly independent, then polynomials $\{f_i\}$ are algebraically independent.
\end{Lem}
\begin{proof}
Please refer to Appendix\ref{pf_lem:independent1}\hspace{-4mm}{\color{white}\scriptsize$\blacklozenge$}\hspace{2.1mm} for the proof.
\end{proof}

\begin{Lem}\emph{(Algebraic Independence is Invariant up to a Random Constant)} \label{lem:independent2}Suppose $\mathcal{K}$ is an algebraically closed field. Polynomials $f_i\in \mathcal{K}[x_1,x_2,...x_{S}]$, $i\in\{1,...,L\}$ are algebraically independent, and $c_i$ are independent random variables drawn from continuous distribution in $\mathcal{K}$. Then $g_i = c_i+f_i$ are algebraically independent almost surely.
\end{Lem}
\proof
Please refer to Appendix\ref{pf_lem:independent2}\hspace{-4mm}{\color{white}\scriptsize$\blacklozenge$}\hspace{2.1mm} for the proof.
\endproof
\begin{Lem}\emph{(Algebraic Independence Leads to Non-empty Vanishing Set)}\label{lem:nonempty} Suppose $\mathcal{K}$ is an algebraically closed field. If polynomials $f_i\in \mathcal{K}[x_1,...,x_{S}]$, $i\in\{1,...,L\}$ are algebraically independent, then the vanishing set of these polynomials, i.e. $\mathcal{V}(f_1,...,f_L)=\{(x_1,...,x_{S}):f_i=0,  i\in\{1,...,L\}\}$ is non-empty.
\end{Lem}
\proof Pleaser refer to Appendix\ref{pf_lem:nonempty}\hspace{-4mm}{\color{white}\scriptsize$\blacklozenge$}\hspace{2.1mm} for the proof.
\endproof

\mysubnote{Main flow of the proof}

 In the following analysis, we prove Thm.~\ref{thm:feasible_s} by applying the new tool developed above. First we transfer the IA problem (Problem~\ref{pro:pIA_c}) into another equivalent form.
\begin{Lem}[Problem Transformation]\label{lem:transform} Problem~\ref{pro:pIA_c} is equivalent to Problem~\ref{pro:pIA2} (defined below) almost surely.
\begin{Prob}[Transformed IA Processing]\label{pro:pIA2} Find $\{\mathbf{U}_k,\mathbf{V}_k\}$ such that
$\mbox{rank}(\mathbf{U}_k)=\mbox{rank}(\mathbf{V}_k)$ $=d_k,$ $\forall k$ and satisfy \eqref{eqn:czero_c}.\end{Prob}
\end{Lem}
\proof
Please refer to Appendix\ref{pf_lem:transform}\hspace{-4mm}{\color{white}\scriptsize$\blacklozenge$}\hspace{2.1mm} for the proof.
\endproof

In Problem~\ref{pro:pIA2}, to ensure that $\mbox{rank}(\mathbf{U}_k)=\mbox{rank}(\mathbf{V}_k)=d_k$, {it is sufficient to assume
that the first $d_k\times d_k$ submatrix of $\mathbf{U}_{k}$, $\mathbf{V}_{k}$, denoted by $\mathbf{U}^{(1)}_k$, $\mathbf{V}^{(1)}_k$, are invertible.} Then we can define $\tilde{\mathbf{U}}_{k}\in\mathbb{C}^{(N_k-d_k)\times d_k}$, $\tilde{\mathbf{V}}_{j}\in\mathbb{C}^{(M_j-d_j)\times d_j}$ as follows:  \begin{eqnarray}\left[\!\!\!\begin{array}{c}\mathbf{I}_{d_k\times d_k}\\ \tilde{\mathbf{U}}_{k}\end{array}\!\!\!\right]
\!=\!\mathbf{U}_k\Big(\mathbf{U}^{(1)}_k\Big)^{\!-1}\!\!\!,\; \left[\!\!\!\begin{array}{c}\mathbf{I}_{d_j\times d_j}\\ \tilde{\mathbf{V}}_{j}\end{array}\!\!\!\right]\!=\!\mathbf{V}_j\Big(\mathbf{V}^{(1)}_j\Big)^{\!-1}\!\!\!.
\label{eqn:tildeuv}
\end{eqnarray} Then \eqref{eqn:czero_c} is transformed into the following form:\footnote{\label{foot:kjpq}Here $k$, $j$, $p$, and $q$ represent the index of Rx, Tx, data stream at Rx side, and data stream at Tx side, respectively. We intensively use this subscript sequence in this paper, e.g. $\mathbf{h}_{kjpq}$, $c^{\mathrm{t}}_{kjpq}$, and $c^{\mathrm{r}}_{kjpq}$.
}
\begin{eqnarray}
\nonumber \!\!f_{kjpq}\!\!\!\!\!\!&\triangleq&\!\!\!\! \! h_{kj\!}(p,\!q)\!+\!\!\!\!\!\sum_{n=1}^{N_k\!-\!d_k}\!\!\!h_{kj\!}(d_k\!+\!n,\!q)
\tilde{u}^H_{k}\!(n,\!p)
\\\nonumber\!\!\!\!\!\!&&\!\!\!\!\!\!+
\sum_{m=1}^{M_j\!-\!d_j}\!\!h_{kj\!}(p,\!d_j\!+\!m)\tilde{v}_{j}\!(m,\!q)
\\\nonumber \!\!\!\!\!\!&&\!\!\!\!\!\!+\sum_{n=1}^{N_k\!-\!d_k}\!\!\sum_{m=1}^{M_j\!\mbox{-}d_j}\!\!h_{kj\!}(d_k\!+\!n,
\!d_j\!+\!m)
\tilde{u}^H_{k}\!(n,\!p)\tilde{v}_{j}\!(m,\!q)
\\\nonumber\!\!\!\!\!\! &=&\!\!\!\!\! \mathbf{h}_{kjpq}\mathbf{v}\!+\!\!\!\!\!\sum_{n=1}^{N_k\!-\!d_k}\!\sum_{m=1}^{M_j\!-\!d_j}
\!\!\!h_{kj\!}(d_k\!+\!n,\!d_j\!+\!m)
\tilde{u}^H_{k}\!(n,\!p)\tilde{v}_{j}\!(m,\!q)\
\\\!\!\!\!\!\! &=&\!\!\!\!\! 0 \label{eqn:czero_poly2}
\end{eqnarray}
where $h_{kj}(p,q)$, $\tilde{u}_{k}(p,q)$, and $\tilde{v}_{j}(p,q)$ are the elements in the $p$-th row and $q$-th column of $\mathbf{H}_{kj}$, $\tilde{\mathbf{U}}_{k}$ and $\tilde{\mathbf{V}}_{j}$, respectively,
\begin{eqnarray}
\nonumber\!\!\mathbf{v}\!\!\!\!\!\!&=&\!\!\!\!\!\!\big[\tilde{u}^H_{1}\!(1,\!1),\tilde{u}^H_{1}\!(2,\!1),...,
\tilde{u}^H_{1}\!(N_{\!1}\!-\!d_1,\!1),\tilde{u}^H_{1}\!(1,\!2),...,
\\\nonumber\!\!&&\!\!\!\!\tilde{u}^H_{1}\!(N_{\!1}\!-\!d_{\!1},\!d_1),
\tilde{u}^H_{2}\!(1,\!1),...,
\tilde{u}^H_{K}\!(N_{\!K}\!-\!d_{\!K},\!d_{\!K}),
\\\nonumber\!\!&&\!\!\!\!\tilde{v}_{1}\!(1,\!1)\:\,,\tilde{v}_{1}\!(2,\!1)\:\,,...,
\tilde{v}_{1}\!(M_{\!1}\!-\!d_{\!1},\!1)\:\,,\tilde{v}_{1}\!(1,\!2)\:\,,...,
\\\!\!&&\!\!\!\!\tilde{v}_{1}\!(M_{\!1}\!-\!d_{\!1},\!d_1)\,\,,
\tilde{v}_{2}\!(1,\!1)\,\,,...,
\tilde{v}_{K}\!(M_{\!K}\!-\!d_{\!K},\!d_{\!K})\!
\big]^T\label{eqn:v}
\end{eqnarray}
and $\{\mathbf{h}_{kjpq}\}$ is the $r$-th row of $\mathbf{H}_{\mathrm{all}}$ defined in Fig.~\ref{fig_Hall}, where $r(k,j,p,q)$ is given by:
\begin{eqnarray}
r(k,\!j,\!p,\!q)\!=\!\!\sum_{k^\dag=1}^{k-1}\!\!\sum_{j^\dag=1\atop\neq k^\dag}^K \!\!d_{k^\dag}d_{j^\dag}\!+\!\!\sum_{j^\dag=1\atop\neq k}^{j-1}\!\!d_{k}d_{j^\dag}\!+\!(p\!-\!1)d_j\!+\!q.\label{eqn:roworder}
\end{eqnarray}

 Substituting \eqref{eqn:czero_poly2} to Lem.~\ref{lem:independent1}{--}\ref{lem:nonempty}, we can prove that Problem~\ref{pro:pIA_c} has solution almost surely if $\mathbf{H}_{\mathrm{all}}$ defined in Fig.~\ref{fig_Hall} is full row rank.
\subsubsection{Proof of Corollary~\ref{cor:confdomin}}
Note that $\mathbf{H}_{\mathrm{all}}\in\mathbb{C}^{C\times V}$, where $C\!=\!\sum_{k=1}^{K}\!\sum_{j=1\atop\neq k}^K d_kd_j$, $ V\!=\!\sum_{k=1}^{K}d_k(M_k+N_k-2d_k)$. $\mathbf{H}_{\mathrm{all}}$ is full row rank if and only if at least one of its $C\!\times\!C$ sub-matrices has non-zero determinant. Therefore, the statement is proved if the following proposition holds:
\begin{Prop} Under a network configuration $\chi$, the determinant of a $C\!\times\!C$ sub-matrix of $\mathbf{H}_{\mathrm{all}}$ is either always zero or non-zero almost surely.\label{prop:det}
\end{Prop}

To prove Prop.~\ref{prop:det}, we first have the following lemma:
\begin{Lem}\label{lem:nonzero}Suppose $x_1,...,x_{S}\in\mathbb{C}$ are independent random variables drawn from continuous distribution, $f$ is a non-constant polynomial $\in\mathbb{C}[x_1,...,x_{S}]$. {Then $f(x_1,...,x_S)\neq 0$ almost surely, i.e. the polynomial evaluated at $(x_1,...,x_S)$ is non zero with probability 1.}
\end{Lem}
\proof When $k=1$, from the Fundamental Theorem of Algebra \cite{Misc:Springer_Plucker}, $f(x_1)=0$ only has finite number of solutions. On the other hand, $x_1$ is drawn from continuous distribution. Hence $f(x_1)\neq0$ almost surely.

For $k\ge2$, the lemma can be proved by using mathematical induction w.r.t. $k$. We omit the details for conciseness.
\endproof

From the Leibniz formula \cite[6.1.1]{Bok:Linear_Meyer}, the determinant of a $C\!\times\!C$ sub-matrix of $\mathbf{H}_{\mathrm{all}}$ can be written as a polynomial $f\in\mathbb{C}(h_{kj}(p,q))$ with no constant term, where $k\neq j\in\{1,...,K\}$, $p\in\{1,...,N_k\}$, $q\in\{1,...,M_j\}$. Further note that the coefficients of $f$ is determined by the configuration of the network $\chi$. Hence, under a {certain} $\chi$, $f$ is either a zero polynomial or a non-constant polynomial. In the latter case, by applying Lem.~\ref{lem:nonzero}, we have that $f\neq 0$ almost surely. This completes the proof.

\mynote{Proof of the new result}
{
\subsubsection{Proof of Corollary~\ref{cor:scaleable}}
As illustrated in Fig.~\ref{fig_scale}, from \eqref{eqn:hu} and \eqref{eqn:hv}, after the scaling, each $\mathbf{H}^{\mathrm{U}}_{kj}$ (or $\mathbf{H}^{\mathrm{V}}_{kj}$) is composed of repeating a submatrix with independent elements $cd_k$ (or $cd_j$) times. Denote the $s$-th time of appearance of this matrix as $\mathbf{H}^{\mathrm{U}}_{kj}(s)$ (or $\mathbf{H}^{\mathrm{V}}_{kj}(s)$). Moreover, we can evenly partition every $\mathbf{H}^{\mathrm{U}}_{kj}(s)$, $\mathbf{H}^{\mathrm{V}}_{kj}(s)$ into $c^2$ independent blocks. Denote the $l$-th diagonal block in $\mathbf{H}^{\mathrm{U}}_{kj}(s)$ (or $\mathbf{H}^{\mathrm{V}}_{kj}(s)$) as $\mathbf{H}^{\mathrm{U}}_{kj}(s,l)$ (or $\mathbf{H}^{\mathrm{V}}_{kj}(s,l)$), $l\in\{1,...,c\}$. Rewritten $\mathbf{H}_{\mathrm{all}}$ as a sum of two matrices, one consists of the diagonal blocks $\{\mathbf{H}^{\mathrm{U}}_{kj}(s,l)\}$, $\{\mathbf{H}^{\mathrm{V}}_{kj}(s,l)\}$  and the other contains the rest of the blocks. Denote the two matrices as $\mathbf{H}^{\mathrm{D}}_{\mathrm{all}}$, $\mathbf{H}^{\tilde{\mathrm{D}}}_{\mathrm{all}}$, respectively.
\begin{figure}[h] \centering
\includegraphics[scale=0.36]{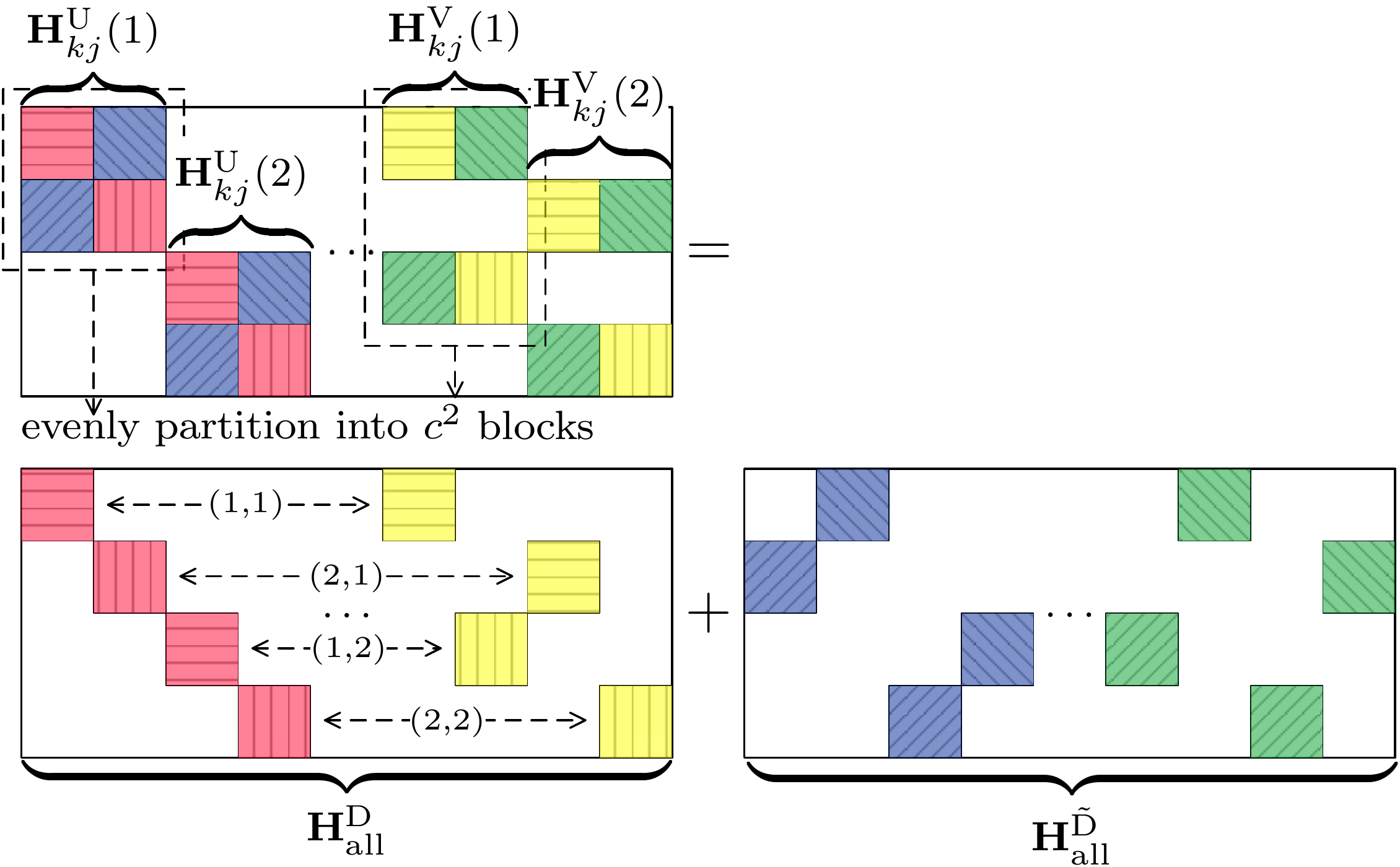}
\caption {Partition of $\mathbf{H}_{\mathrm{all}}$. In this figure, $d_k=d_j=1$, $c=2$.}
\label{fig_scale}
\end{figure}

Since $\mathbf{H}^{\mathrm{D}}_{\mathrm{all}}$, $\mathbf{H}^{\tilde{\mathrm{D}}}_{\mathrm{all}}$ are independent, it is sufficient to show that $\mathbf{H}^{\mathrm{D}}_{\mathrm{all}}$ is full row rank. As illustrated in Fig.~\ref{fig_scale}, by combining $\{\mathbf{H}^{\mathrm{U}}_{kj}(s,l):s\in\{(l'\!-\!1)d_k\!+\!1,(l'\!-\!1)d_k\!+\!2, ...,l'd_k\}\}$ with $\{\mathbf{H}^{\mathrm{V}}_{kj}(s',l'):s'\in\{(l\!-\!1)d_j\!+\!1,(l\!-\!1)d_j\!+\!2,...,ld_k\}\}$, $l,l'\in\{1,...,c\}$, we obtain $c^2$ combinations. By collecting the blocks with the same combination index $(l,l')$ in different $\mathbf{H}^{\mathrm{U}}_{kj}$ and $\mathbf{H}^{\mathrm{V}}_{kj}$, we obtain $c^2$ submatrices identical to the $\mathbf{H}_{\mathrm{all}}$ before scaling. Since these submatrices are full rank almost surely and are on different rows and columns of $\mathbf{H}^{\mathrm{D}}_{\mathrm{all}}$, $\mathbf{H}^{\mathrm{D}}_{\mathrm{all}}$ is full rank almost surely. This completes the proof.
}

\mynote{New proof to eliminate the bug.}
\subsubsection{Proof of Corollary~\ref{cor:sym}}
\mysubnote{Highlight the outline of the proof.}

 For notational convenience, we will use notation $(M\times N, d)^K$ to represent the configuration of a symmetric MIMO interference network, where the meaning of the letters are the as same those in Cor.~\ref{cor:sym}.

 The ``only if" side can be easily derived from \eqref{eqn:f3}. We adopt the following procedures to prove the ``if" side:
 \begin{itemize}
 \item[A.] Construct one special category of channel state $\{\mathbf{H}_{kj}\}$.
 \item[B.] Show that $\mathbf{H}_{\mathrm{all}}$ is full rank almost surely under the special category of channel state.
 \item[C.] From Cor.~\ref{cor:confdomin}, if Procedure B is completed, $\mathbf{H}_{\mathrm{all}}$ is full rank almost surely and hence we prove the corollary.
\end{itemize}

\mysubnote{Details of the proof.}

Now we start the detailed proof following the outline illustrated above. We first have two lemmas.
\begin{Lem}[Sufficient Condition for Full Rankness]\label{lem:vdomin} Denote $\mathcal{H}^{\mathrm{V}}_k=\mbox{span}(\{\mathbf{h}_{kjpq}:j\in\{1,...,K\}, j\neq k, p,q\in\{1,...,d\}\})\cap\mathcal{V}$, where $\mathcal{V}=\mbox{span}\bigg(\bigg[\!\!\begin{array}{c}\mathbf{0}_{D^{\mathrm{U}}\!\!
\times\! D^{\mathrm{V}}}\\ \mathbf{I}_{D^{\mathrm{V}}\!\!\times\! D^{\mathrm{V}}}\end{array}\!\!\bigg]\bigg)$, $D^{\mathrm{U}}=K\!(N\!-\!d)d$, $D^{\mathrm{V}}=K\!(M\!-\!d)d$. When $N\ge 2d$, $\mathbf{H}_{\mathrm{all}}$ is full row rank almost surely if the basis vectors of all $\mathcal{H}^{\mathrm{V}}_k$, $k\in\{1,...,K\}$ are linearly independent.
\end{Lem}

\proof Please refer to Appendix\ref{pf_lem:vdomin}\hspace{-3.8mm}{\color{white}\scriptsize$\blacklozenge$}\hspace{1.9mm} for the proof.
\endproof

\begin{Lem}[Full Rankness of Special Matrices]\label{lem:structure}
A matrix $\mathbf{H}_{\mbox{\scriptsize sub}}$ with the following structure is full rank almost surely.
\begin{itemize}
\item[S1.]$\mathbf{H}_{\mbox{\scriptsize sub}}$ is composed of $d\times d$ blocks, each block is composed of $K\times K$ sub-blocks, and each sub-block is aggregated by $M-d$ number of $1\times M-d$ vectors. Matrix  $\textcircled{\scriptsize 1}$ in Fig.~\ref{fig_symid_sub} illustrates an example with $d=2$, $K=2$, $M=4$.
\item[S2.]Denote the sub-blocks as $\mathbf{B}_{ss'kk'}$, $s,s'\in\{1,...,d\}$, $k,k'\in\{1,...,K\}$, where $s$, $s'$ denote the vertical and horizontal position of the block, and $k$, $k'$ denote the vertical and horizontal position of the sub-block within the block (e.g. $\mathbf{B}_{1211}$ in Fig.~\ref{fig_symid_sub}). All diagonal blocks are block-diagonal, i.e. $\mathbf{B}_{sskk'}=\mathbf{0}$, if $k\neq k'$. Denote the $k$-th diagonal sub-block in block $s$ as $\mathbf{B}^{\mathrm{D}}_{sk}$ (e.g. $\mathbf{B}^{\mathrm{D}}_{12}$ in Fig.~\ref{fig_symid_sub}).
\item[S3.]   The elements in every $\mathbf{B}^{\mathrm{D}}_{sk}$ are independent random variables drawn from continuous distribution.
\item[S4.] $\mathbf{B}^{\mathrm{D}}_{sk}$ is independent of all the diagonal sub-blocks with different sub-block index, i.e. $\{\mathbf{B}^{\mathrm{D}}_{s'k'}:k'\neq k\}$ and all the sub-blocks in the same columns and rows, i.e. $\{\mathbf{B}_{ss'kk'}, \mathbf{B}_{s'sk'k}:s'\neq s\,\mbox{or}\, k'\neq k\}$. The vectors in the off-diagonal blocks are either $\mathbf{0}$, or independent of \emph{all} diagonal sub-blocks, or repetition of a certain vector in the diagonal sub-blocks (positioned in different columns).
\item[S5.]Define diagonal sub-blocks $\mathbf{B}^{\mathrm{D}}_{sk}$ and $\mathbf{B}^{\mathrm{D}}_{s'k'}$ are \emph{associated}, if a certain vectors in sub-blocks $\mathbf{B}_{ss'kk'}$ or $\mathbf{B}_{s'sk'k}$, $s\neq s'$, $k\neq k'$ is a repetition of a certain vector in the diagonal sub-blocks. Each diagonal sub-block $\mathbf{B}^{\mathrm{D}}_{sk}$ is associated with at most one diagonal sub-block in the neighboring blocks with different sub-block index, i.e.  $\mathbf{B}^{\mathrm{D}}_{(\!s\!-\!1\!)k'}$ and $\mathbf{B}^{\mathrm{D}}_{(\!s\!+\!1\!)k''}$, for some $k',k''\neq k$. Note that when $s=1$ or $d$, each diagonal sub-block is associated with at most one sub-block.
\end{itemize}
\begin{figure}[h] \centering
\includegraphics[scale=0.34]{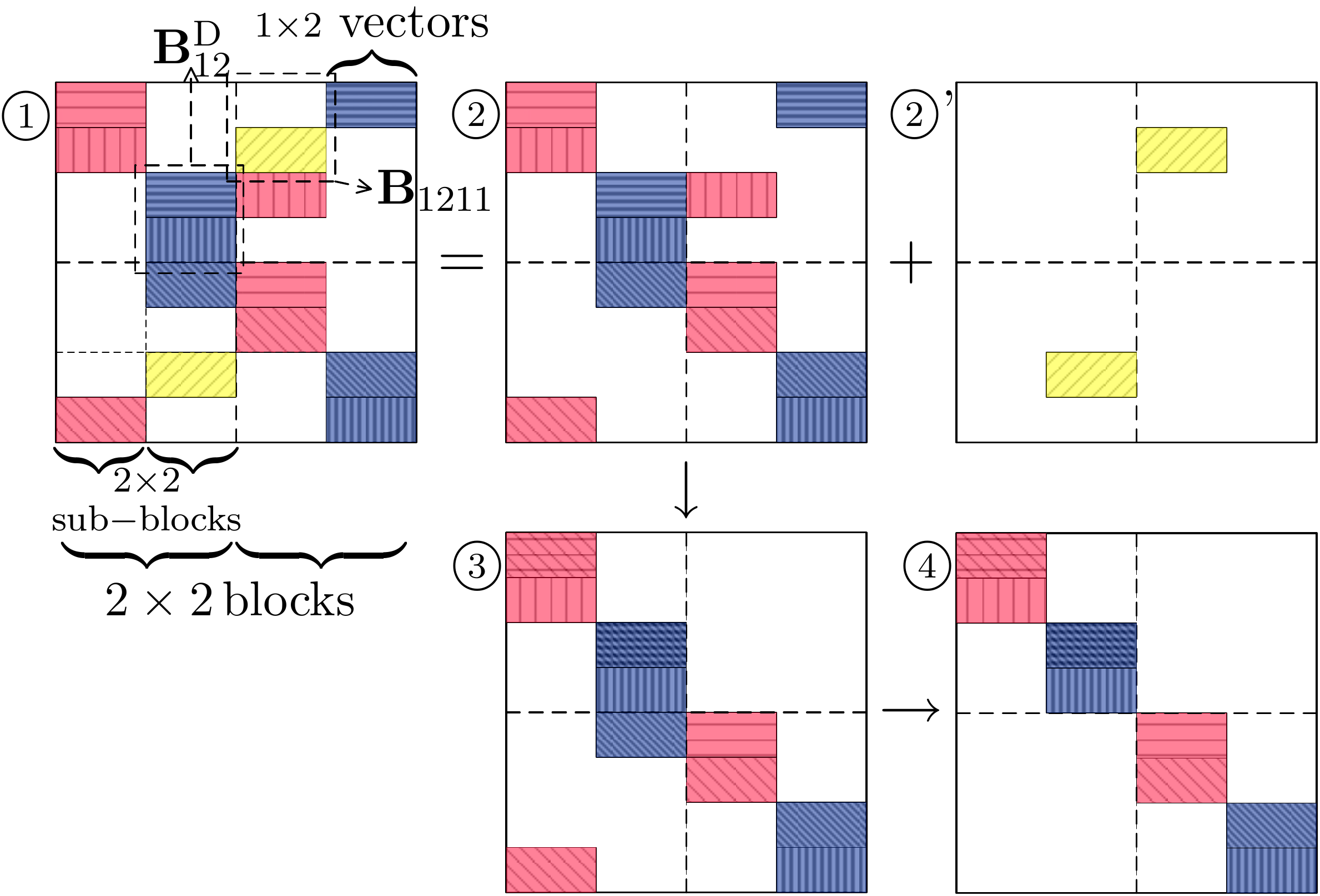}
\caption {Outline of the proof of Lem.~\ref{lem:structure}.}
\label{fig_symid_sub}
\end{figure}
\end{Lem}
\proof
Please refer to Appendix\ref{pf_lem:structure}\hspace{-3.7mm}{\color{white}\scriptsize$\blacklozenge$}\hspace{1.8mm} for the proof.
\endproof

Now we start the main procedures of the proof. We first narrow down the scope:
\begin{itemize}
\item[-]When $K=2$, the proof is straightforward.
\item[-]If the corollary is true in the boundary cases, i.e. $M+N=(K+1)d$, it is true in general.
\item[-]With Cor.~\ref{cor:scaleable}, it is sufficient to consider the case in which $\mbox{gcd}(d,N,M)=1$. In the boundary cases, since $d|(M\!+\!N)$, $\mbox{gcd}(d,N,M)=1\Rightarrow\mbox{gcd}(d,M)=1$.
\item[-]If $d=1$, the corollary is reduced to a special case of Cor.~\ref{cor:div}.
\end{itemize}
Hence, we focus on cases in which $K\ge 3$, $M+N=(K+1)d$, $\mbox{gcd}(d,M)=1$, and $d\ge2$. To improve readability of the proof, we adopt a $(7\times 8, 3)^4$ network as an example. From Fig.~\ref{fig_Hall}, matrix $\mathbf{H}_{\mathrm{all}}$ of the example network is given by the first matrix\footnote{Note that here the value of the submatrices $\mathbf{H}^{U}_{kj}$ are specified. We will explain how we construct this specification later.} in Fig.~\ref{fig_symid}.

\begin{figure*}[t] \centering
\includegraphics[scale=0.95]{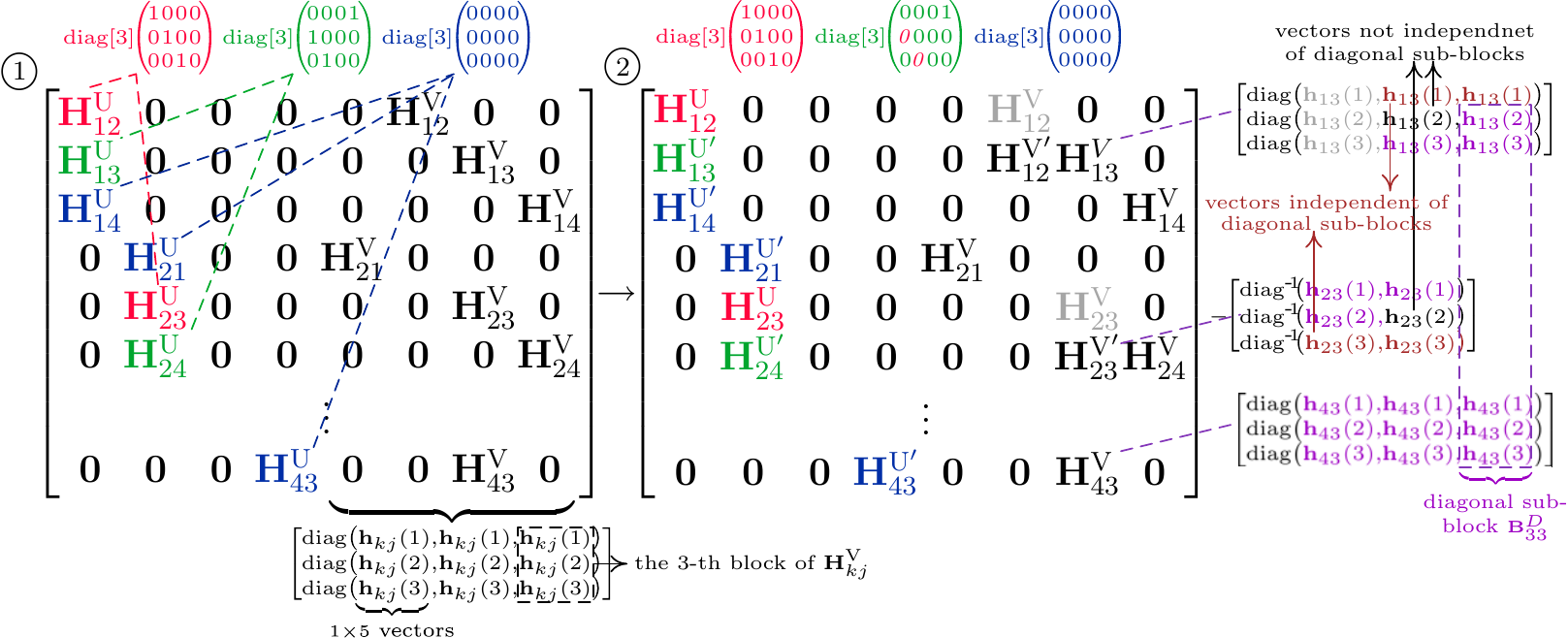}
\caption {Illustration of proving that $\mathbf{H}_{\mathrm{all}}$ is full rank in a $(7\times 8, 3)^4$ network.}
\label{fig_symid}
\end{figure*}

\mysubnote{Construct a specific channel state.}
\mysubnote{It take me a lot of effort to make all ends meet.}

\begin{itemize}
\item[A.] Specify $\{\mathbf{H}^{\mathrm{U}}_{kj}\}$ as in Fig.~\ref{fig_specifyhu}, in which
\begin{eqnarray}
\!\!\!\!\!\!\!P(k,j)\!\!\!\!\!&=&\!\!\!\!\!\mbox{mod}\big(d(\mbox{mod}(j\!-\!k,K)\!-\!1),
N\!-\!d\big)\label{eqn:pos}\\
\!\!\!\!\!\!\!R(k,j)\!\!\!\!\!&=&\!\!\!\!\!\!\left\{\!\!\!\!\begin{array}{ll}
d &\!\!\!\!\mbox{if mod}(j\!-\!k,K)\!\le\! \lfloor\!\frac{N}{d}\!\rfloor,\\
\mbox{mod}(N\!-\!1,d)&\!\!\!\!\mbox{if }\mbox{mod}(j\!-\!k,K)\!=\! \lfloor\!\frac{N}{d}\!\rfloor\!+\!1,\\
0 &\!\!\!\!\mbox{otherwise.}
\end{array}\right.\label{eqn:rownumber}
\end{eqnarray}
\begin{figure}[h] \centering
\includegraphics[scale=0.8]{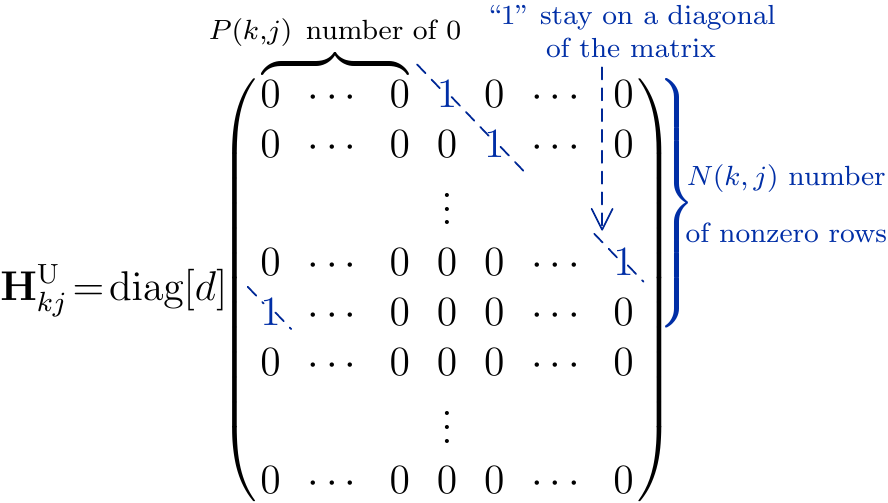}
\caption {Specify $\{\mathbf{H}^{\mathrm{U}}_{kj}\}$.}
\label{fig_specifyhu}
\end{figure}

Matrix  $\textcircled{\scriptsize 1}$ in Fig.~\ref{fig_symid} serves as an example of this specification. Both \eqref{eqn:pos}, \eqref{eqn:rownumber} are cyclic symmetrical w.r.t. user indices $k,j$, i.e. index pairs $(k,j)$ and $(\mbox{mod}(k+\delta),\mbox{mod}(j+\delta))$ lead to the same $P(k,j)$ and $R(k,j)$, $\forall \delta\in\mathbb{Z}$. This property will help us to exploit the symmetry of the network configuration in the proof.
\item[B.] From \eqref{eqn:hv}, each $\mathbf{H}^{\mathrm{V}}_{kj}$ consists of $d$ independent $1\times (M-d)$ vectors repeating for $d$ times. For notational convenience, denote these vectors as $\mathbf{h}_{kj}(1)\sim\mathbf{h}_{kj}(d)$ and denote their $s$-th time of appearance as the $s$-th block of $\mathbf{H}^{\mathrm{V}}_{kj}$. The small matrix below matrix $\textcircled{\scriptsize 1}$ in Fig.~\ref{fig_symid} has given such an example. As illustrated by matrix $\textcircled{\scriptsize 2}$ in Fig.~\ref{fig_symid}, under the specification in Fig.~\ref{fig_specifyhu}, we can adopt row operations to remove the ``1"s that reappear in the same columns. From Lem.~\ref{lem:vdomin}, it is sufficient to prove that the row vectors which are occupied by the $s$-th block of $\mathbf{H}^{\mathrm{V}}_{kj}$ are linearly independent, where $s$, $k$, and $j$ satisfy:
    \begin{eqnarray}
    s\!\!\!\!\!&\in&\!\!\!\!\!\!\left\{\!\!\!\!\begin{array}{ll}
    \emptyset &\!\!\!\!\mbox{if mod}(j\!-\!k,K)\!\le\! \lfloor\!\frac{N}{d}\!\rfloor\!-\!1,\\
    \{\mbox{mod}(N,d)\!+\!1,...,d\} &\!\!\!\!\mbox{if }\mbox{mod}(j\!-\!k,K)
    \!=\! \lfloor\!\frac{N}{d}\!\rfloor,\\
    \{1,...,d\} &\!\!\!\!\mbox{otherwise.}
    \end{array}\right.\label{eqn:nullnumber}
    \end{eqnarray}
    Also note that after the row operation, the $1\sim (d-1)$-th block of $\mathbf{H}^{\mathrm{V}}_{kj}$, $k=\mbox{mod}(j-2,K)+1$ is replicated, taken a minus sign and moved to other rows. Denote these new submatrices as $\{\mathbf{H}^{\mathrm{V}'}_{kj}\}$, $k=\mbox{mod}(j-2,K)+1$, $j\in\{1,...,K\}$. Now we can adopt Lem.~\ref{lem:structure} to prove the linear independence of the row vectors specified by \eqref{eqn:nullnumber}. Specifically, for every $j\in\{1,...,K\}$, select the following vectors:
    \begin{itemize}
    \item \emph{When} $k=\mbox{mod}(j-2,K)+1$\emph{:} $\mathbf{h}_{kj}(1)\sim\mathbf{h}_{kj}(d-s)$ in the $s$-th block of $\mathbf{H}^{\mathrm{V}'}_{kj}$, where $s\in\{1,...,d-1\}$.
    \item \emph{When} $k=\mbox{mod}(j-\lfloor\frac{N}{d}\rfloor-1,K)+1$\emph{:}
        $\mathbf{h}_{kj}(d^\dag)\sim\mathbf{h}_{kj}(d)$ in the $s$-th block of $\mathbf{H}^{\mathrm{V}'}_{kj}$, where $d^\dag=d+\mbox{mod}(N,d)+1-s$, $s\in\{\mbox{mod}(N,d)+1,...,d\}$.
    \item \emph{When} $k=\mbox{mod}(j-\lfloor\frac{N}{d}\rfloor-2,K)+1$\emph{:}
        $\mathbf{h}_{kj}(d^\dag)\sim\mathbf{h}_{kj}(d)$ in the $s$-th block of $\mathbf{H}^{\mathrm{V}'}_{kj}$, where $d^\dag=\max\big(1,\mbox{mod}(N,d)+1-s\big)$, $s\in\{1,...,d\}$.
    \item \emph{When} $k=\mbox{mod}(j-l,K)+1$\emph{:} $l\in\{\lfloor\frac{N}{d}\rfloor+3,...,K\}$: all vectors in $\mathbf{H}^{\mathrm{V}}_{kj}$.
    \end{itemize}
    Then as illustrated by the small matrices on the right side of Fig.~\ref{fig_symid}, by adopting this selection mechanism, we have chosen $Kd(M-d)$ vectors, which form $Kd$ number of $(M-d)\times(M-d)$ submatrices positioned on different rows and columns. Denote these submatrices as $\mathbf{B}_{sj}$, where $s\in\{1,...,d\}$ and $j\in\{1,...,K\}$ represent the block and user index, respectively. Map $\mathbf{B}_{sj}$ to $\mathbf{B}^{\mathrm{D}}_{s'j}$ in Lem.~\ref{lem:structure}, where $s$ and $s'$ satisfy: $s= \mbox{mod}\big(d+(d-s')\mbox{mod}(M,d)-1,d\big)+1$. As $\mbox{gcd}(M,d)=1$, this is a one to one mapping. Since the submatrices are positioned on different rows and columns, we can move them to the principle diagonal and verify that the structures required in Lem.~\ref{lem:structure} are satisfied. This completes the proof.
\end{itemize}

\subsubsection{Proof of Corollary~\ref{cor:div}}
We first prove some key lemmas and then turn to the main flow of the proof.
\begin{Lem}[Sufficient Conditions for IA Feasibility]\label{lem:suf1}  If there exists a set of binary variables $\{c^{\mathrm{t}}_{kjpq},c^{\mathrm{r}}_{kjpq}\in\{0,1\}\}$, $k,j\in\{1,...,K\}$, $k\neq j$, $p\in \{1,...,d_{k}\}$, $q\in \{1,...,d_{j}\}$ that satisfy the following constraints, Problem~\ref{pro:pIA_c} has solution almost surely.
\begin{eqnarray}
&&c^{\mathrm{t}}_{kjpq}+c^{\mathrm{r}}_{kjpq}=1,\label{eqn:cv1}
\\&&\sum_{j=1,\neq k}^K \sum_{q=1}^{d_j} c^{\mathrm{r}}_{kjpq}\le N_k-d_k,\; \forall k, \label{eqn:cv2r}
\\&&\sum_{k=1,\neq j}^K \sum_{p=1}^{d_k} c^{\mathrm{t}}_{kjpq}\le M_j-d_j,\; \forall j, \label{eqn:cv2t}
\\\nonumber && c^{\mathrm{t}}_{kj1q}=...=c^{\mathrm{t}}_{kjd_kq},\; \forall k,j,q, \mbox{ OR }
\\&&c^{\mathrm{t}}_{kjp1}=...=c^{\mathrm{t}}_{kjpd_j},\; \forall k,j,p. \label{eqn:cv3}
\end{eqnarray}
\end{Lem}
\proof Please refer to Appendix\ref{pf_lem_cor:sym}\hspace{-4mm}{\color{white}\scriptsize$\blacklozenge$}\hspace{2.1mm} for the proof.
\endproof
\mysubnote{Explain the physical meaning of the Lemma.}
\begin{Remark}[Interpretation of $c^{\mathrm{t}}_{kjpq},c^{\mathrm{r}}_{kjpq}$]
The binary variables $c^{\mathrm{t}}_{kjpq},c^{\mathrm{r}}_{kjpq}$ represent a \emph{constraint allocation} policy. An IA constraint $f_{kjpq}=0$ (defined in \eqref{eqn:czero_poly2}) can be assigned to transceivers with non-zero coefficients in $f_{kjpq}$, i.e. $\tilde{\mathbf{u}}_{kp}\in\mathbb{C}^{1\!\times\!(N_k\!-\!d_k)}$ or $\tilde{\mathbf{v}}_{jq}\in\mathbb{C}^{1\!\times\!(N_j\!-\!d_j)}$. Here $\tilde{\mathbf{u}}_{kp}$ $(\tilde{\mathbf{v}}_{jq})$ denotes the $p$ $(q)$-th column of $\tilde{\mathbf{U}}_{k}$ $(\tilde{\mathbf{V}}_{j})$. $c^{\mathrm{r}}_{kjpq}=1$ ($c^{\mathrm{t}}_{kjpq}=1$) means that the IA constraint $f_{kjpq}=0$ is assigned to the decorrelator (precoder) for the $p$ ($q$)-th stream at Rx $k$ (Tx $j$).~\hfill\IEEEQED
\end{Remark}
{
\begin{Remark}[Meaning of Constraints in Lem.~\ref{lem:suf1}]\label{rem:suf1}
\begin{itemize}
\item[]
\item {\eqref{eqn:cv1}:} Each IA constraint $f_{kjpq}=0$ is assigned once and only once.
\item {\eqref{eqn:cv2r}:} The total number of constraints assigned to the decorrelator of any stream, i.e. $\tilde{\mathbf{u}}_{kp}$ is no more than the length of this decorrelator, i.e. $\mbox{size}\big(\tilde{\mathbf{u}}_{kp}\big)=N_k-d_k$.
\item {\eqref{eqn:cv2t}:} The dual version of \eqref{eqn:cv2r}.
\item {\eqref{eqn:cv3}:} The constraint assignment policy $\{c^{\mathrm{r}}_{kjpq},c^{\mathrm{t}}_{kjpq}\}$ is symmetric w.r.t. Rx side stream index $p$ or Tx side stream index $q$.
~\hfill~\IEEEQED
\end{itemize}
\end{Remark}
}
The following lemma illustrate the relation between the \emph{sufficient} conditions proposed in Lem.~\ref{lem:suf1} and the \emph{necessary} conditions proposed in Thm.~\ref{thm:feasible}.

\begin{Lem}[Necessary Conditions of IA Feasibility]\label{lem:nes1}
A network configuration $\chi$ satisfies the necessary feasibility condition \eqref{eqn:f3}, if and only if there exists a set of binary variables $\{c^{\mathrm{t}}_{kjpq},c^{\mathrm{r}}_{kjpq}\in\{0,1\}\}$, $k,j\in\{1,...,K\}$, $k\neq j$, $p\in \{1,...,d_{k}\}$, $q\in \{1,...,d_{j}\}$ that satisfy \eqref{eqn:cv1}{--}\eqref{eqn:cv2t}.
\end{Lem}
\proof Please refer to Appendix\ref{pf_lem:nes1}\hspace{-4mm}{\color{white}\scriptsize$\blacklozenge$}\hspace{2.1mm} for the proof.
\endproof
\mysubnote{Describe the contribution of the lemmas.}
\begin{Remark}[Insight of Lem.~\ref{lem:suf1}, ~\ref{lem:nes1}]
Prior works studying the IA feasibility problem on MIMO interference networks have shown  that the \emph{properness}\footnote{This terminology is first defined in \cite{JafarDf}, which means the number of the free variables in transceiver design must be no less than the number of the IA constraints.} condition, i.e. \eqref{eqn:f3}, is the major factor that characterizes the IA feasibility conditions. However, \eqref{eqn:f3} contains $\mathcal{O}(2^{K^2})$ number of correlated inequalities. Such a complicated condition is hard to trace in both analysis and practice.

Lem.~\ref{lem:nes1} enables us to significantly simplify \eqref{eqn:f3}. By exploiting the idea of \emph{constraint allocation}, Lem.~\ref{lem:nes1} converts \eqref{eqn:f3} to \eqref{eqn:cv1}{--}\eqref{eqn:cv2t}, which consist of only $\mathcal{O}(K)$ number of constraints. Lem.~\ref{lem:suf1} shows that with an additional requirement \eqref{eqn:cv3} on the constraint allocation policy $\{c^{\mathrm{t}}_{kjpq},c^{\mathrm{r}}_{kjpq}\}$, the IA feasibility is guaranteed.
~\hfill~\IEEEQED
\end{Remark}

Now we turn to the main flow of the proof {for Cor.~\ref{cor:div}}. The ``only if" side is directly derived from \eqref{eqn:f3}. The ``if" side is completed by adopting Lem.~\ref{lem:suf1}. Please refer to $\mbox{Appendix\ref{pf_constructc}\hspace{-2.6mm}\color{white}\scriptsize$\blacklozenge$\hspace{1.5mm}}$ for the details of constructing $\{c^{t*}_{kjpq}, c^{r*}_{kjpq}\}$ which satisfy \eqref{eqn:cv1}{--}\eqref{eqn:cv3}.

\section{Summary and Future Work}
\label{sec:conclude}
\mynote{Summarize the contribution of this work.}

This work further consolidates the theoretical basis of IA. We have proved a sufficient condition of IA feasibility which applies to MIMO interference networks with general configurations {and discovered that IA feasibility is preserved when scaling the network.} Further analysis show that the sufficient condition and the necessary conditions coincide in a wide range of network configurations and provide some simple analytical conditions. These results unify and extend the pertaining theoretical works {in the literature} \cite{JafarDf,Misc:IAfeasible_Tse,Conf:IAfeasible_Luo,Jul:IAfeasible_Luo} and facilitate future analysis on MIMO interference networks.

\mynote{Point out the remaining issues.}

Despite the progress made in the prior works and this work, the issue of IA feasibility is yet not fully solved. In particular, there may be gaps between the necessary conditions in Thm.~\ref{thm:feasible} and the sufficient condition in Thm.~\ref{thm:feasible_s} and therefore the exact feasibility conditions of IA are still not determined in general. Merging the gap between the necessary and the sufficient side shall be the direction for future works.

\appendices
\section*{Appendices}
\addcontentsline{toc}{section}{Appendices}
\label{sec:appendix}
\subsection{Proof of Lemma~\ref{lem:independent1}}
\label{pf_lem:independent1}
When vectors $\{\mathbf{h}_i\}$ are linearly independent, we have that $L\le S$. Note that if the lemma holds when $L=S$, it must hold in general. We will use contradiction to prove the statement. Suppose $f_i$,  $i\in\{1,2,...,L\}$ are algebraically dependent. Then from the definition, there must exist a nonzero polynomial $p$, such that $p(f_1,f_2,...,f_S)= 0$. Without loss of generality, denote $p = p_0 +p_1 +...+p_D$, where $p_d$ contains all the $d$-th degree terms in $p$, $D\in\mathbb{N}$. Then we have:
\begin{eqnarray}
\nonumber p(f_1,f_2,...,f_S)=   p_0\!+\!p_1(\sum_{j=1}^{S}h_{1j}x_j,...,\sum_{j=1}^{S}h_{Sj}x_j)+
\\\Big(p_1(g_1,...,g_S)\!+\!\!\!\sum_{d=2}^{D}p_d(f_1,f_2,...,f_S)\Big)
\!=\! 0\label{eqn:pzero}
\end{eqnarray}

Note that all the terms in $(p_1(g_1,...,g_S)+\sum_{d=2}^{D}p_d(f_1,f_2,...,f_S))$ have degree no less than 2, from \eqref{eqn:pzero}, $p_1(\sum_{j=1}^{S}h_{1j}x_j,...,\sum_{j=1}^{S}h_{Sj}x_j)=0$, $\forall x_1,...,x_j\in\mathcal{K}$. Denote $y_i=\sum_{j=1}^{S}h_{ij}x_j$, we have that:
\begin{eqnarray}p_1(y_1,...,y_S)=0.\label{eqn:yzero}
\end{eqnarray}

 Note that the coefficient vectors $\mathbf{h}_i$ are linearly independent, $ [x_1,...,x_S]\rightarrow[y_1,...,y_S]$ is a bijective linear map. Therefore, $\{y_1,...,y_S\}\cong\{x_1,...,x_S\}=\mathcal{K}^S$. Hence, from \eqref{eqn:yzero}, we have that $\mathcal{V}(p_1)=\{y_1,...,y_S\}\cong\mathcal{K}^S$, which means $p_1$ must be a zero polynomial.

Similarly, when $p_1$ is a zero polynomial, by analyzing the coefficients of the second order terms, we have that $\mathcal{V}(p_2)\cong\mathcal{K}^S$ and therefore $p_2$ is also a zero polynomial. By using mathematical induction, we can show that $p_1, p_2,...p_D$ are zero polynomials and hence $p$ a zero polynomial, which is a contradiction with the assumption that $f_i$,  $i\in\{1,2,...,L\}$ are algebraically dependent. This completes the proof.

\subsection{Proof of Lemma~\ref{lem:independent2}}
\label{pf_lem:independent2}
We first prove that $g_1, f_2,...,f_L$ are algebraically independent. Then the Lemma can be proved by repeating the same trick $L$ times.
We will use contradiction to prove the statement. Suppose $g_1, f_2,...,f_L$ are algebraically dependent, i.e. {there exists} a non zero polynomial $p$ such that $p(g_1,f_2,...,f_L)=0$. {Without loss of generality}, denote $p(g_1,f_2,...,f_L)=\sum_{d=0}^{D}g^d_1p_d(f_2,...,f_L)$, $D\in\mathbb{N}$, where $p_d$ is a polynomial function of $f_2\sim f_L$. {Then we can define polynomial $p^\dag(f_1,f_2,...,f_L)$:
\begin{eqnarray}
&&\!\!\!\!\!\!\!\!\!\!\!\!\!\!\!\!\nonumber p^\dag(f_1,f_2,...,f_L)\triangleq p(f_1+c_1,f_2,...,f_L)=\\
&&\!\!\!\!\!\!\!\!\!\!\!\!\!\!\!\!\quad\sum_{d=0}^{D}f^d_1p_d+c_1   \sum_{d=0}^{D}\sum_{s=0}^{d-1}\Big( {s\atop d}\Big)c_1^{(d\mbox{-}1)}f^{(d\mbox{-}s)}_1p_{d}(f_2,...,f_L)
\end{eqnarray}
where $\Big( {s\atop d}\Big)$ denotes the number of  $s$-combination of a set with $d$ elements.} Note that $\sum_{d=1}^{D}f^d_1p_d=p(f_1,...,f_L)$ is nonzero, and $c_1$ is independent of the coefficients in $p_D$, we have that $p^\dag$ is nonzero almost surely. However, $p^\dag(f_1,f_2,...,f_L)=0$, which contradicts with the assumption that $f_1,...,f_L$ are algebraically independent. This completes the proof.

\subsection{Proof of Lemma~\ref{lem:nonempty}}
\label{pf_lem:nonempty}
Since $\{f_i\}$ are algebraically independent, from \cite[Thm.0.4, Lecture 2]{Misc:AdvAlgebra},  $\mathcal{K}[f_1,...,f_L]\cong \mathcal{K}[y_1,...,y_L]$, {where $y_1,...,y_L$ are variables in $\mathcal{K}$}. Hence, $\langle f_1,...,f_L \rangle \cong \langle y_1,...,y_L\rangle$, where $\langle z_1,...,z_L\rangle$ denotes the ideal generated by $z_1,...,z_L$. Note that ideal $\langle z_1,...,z_L\rangle$ is \emph{proper}, i.e. does not contain $1$, so is ideal $\langle f_1,...,f_L \rangle$. From Hilbert's Nullstellensatz Theorem \cite[Thm. 3.1, Chap. I]{Bok:ComA_Kunz}, $\mathcal{V}(f_1,...f_L)$ is non-empty.

\subsection{Proof of Lemma~\ref{lem:transform}}
\label{pf_lem:transform}
Firstly, it is easy to see that a solution of Problem~\ref{pro:pIA_c} is a solution of Problem~\ref{pro:pIA2}. Conversely, since the channel state of the direct links $\{\mathbf{H}_{kk}\}$ are full rank with probability 1 and are independent of that of cross links $\{\mathbf{H}_{kj}\},$ $k\neq j$, a solution of  Problem~\ref{pro:pIA2} is also a solution of Problem~\ref{pro:pIA_c} with probability 1.

\subsection{Proof of Lemma~\ref{lem:vdomin}}
\label{pf_lem:vdomin}
We first have two lemmas.
\begin{Lem2}\hspace{-9.2mm}{\color{white}\scriptsize$\blacklozenge$}\hspace{6.5mm} In $\mathbf{H}_{\mathrm{all}}$, the row vectors that are related to a same Rx are linearly independent almost surely, i.e. for every $k\in\{1,...,K\}$, vectors $\mathbf{h}_{kjpq}$, $j\in\{1,...,K\}, j\neq k, p,q\in\{1,...,d\}$, are linearly independent almost surely.\label{lem:samerxid}
\end{Lem2}
\proof From \eqref{eqn:hu}, and the fact $M-d\ge d$, we have that every submatrix $\mathbf{H}^{\mathrm{V}}_{kj}$ is full rank almost surely. Since for a given $k$, submatrices $\mathbf{H}^{\mathrm{V}}_{kj},$ $j\in\{1,...,K\}, j\neq k$ position on different rows and columns in $\mathbf{H}_{\mathrm{all}}$, the lemma is proved.
\endproof

\begin{Lem2}\hspace{-9.2mm}{\color{white}\scriptsize$\blacklozenge$}\hspace{6.5mm} As illustrated in Fig.~\ref{fig_symid}, denote $\mathbf{h}^{\mathrm{U}}_{kjpq}$ as the vector consists of the first $K\!(N\!-\!d)d$ elements of $\mathbf{h}_{kjpq}$. For all $k,j, k^\dag,j^\dag\in\{1,...,K\}$, $k\neq k^\dag$, $p,p^\dag\in\{1,...,d_k\}$, $q,q^\dag\in\{1,...,d_j\}$:
$\mathbf{h}^{\mathrm{U}}_{kjpq}\perp\mathbf{h}^{\mathrm{U}}_{k^\dag j^\dag p^\dag q^\dag}$.\label{lem:uorth}
\end{Lem2}
\proof Straight forward from the structure of $\mathbf{H}_{\mathrm{all}}$.
\endproof

We will prove the lemma by proving its converse-negative proposition. From Lem.\ref{lem:samerxid}\hspace{-6.2mm}{\color{white}\scriptsize$\blacklozenge$}\hspace{4mm}, if $\mathbf{H}_{\mathrm{all}}$ is not full row rank, there must exists a non-zero vector $\mathbf{h}$ and set $\mathcal{A}, \mathcal{B}\subset\{1,...,K\}$, $\mathcal{A}\cap\mathcal{B}=\emptyset$ such that
\begin{eqnarray}
\mathbf{h}\in\big(\cup_{k\in\mathcal{A}}\mathcal{H}_k\big)
\cap\big(\cup_{k\in\mathcal{B}}\mathcal{H}_k\big),\label{eqn:hintersect}
\end{eqnarray}
where $\mathcal{H}_k=\mbox{span}(\{\mathbf{h}_{kjpq}:j\in\{1,...,K\}, j\neq k, p,q\in\{1,...,d\}\})$. Furthermore, from Lem.\ref{lem:uorth}\hspace{-6.2mm}{\color{white}\scriptsize$\blacklozenge$}\hspace{4.5mm} and the fact that $\mathcal{A}\cap\mathcal{B}=\emptyset$, we have that the first $K\!(N\!-\!d)d$ elements of $\mathbf{h}$ must be 0. By combining this fact with \eqref{eqn:hintersect}, we have that:
\begin{eqnarray}
\mathbf{h}\in\big(\cup_{k\in\mathcal{A}}\mathcal{H}^{\mathrm{V}}_k\big)
\cap\big(\cup_{k\in\mathcal{B}}\mathcal{H}^{\mathrm{V}}_k\big),\label{eqn:hintersect}
\end{eqnarray}
which means the basis vectors of $\mathcal{H}^{\mathrm{V}}_k$, $k\in\mathcal{A}\cup\mathcal{B}$ are linearly dependent. This completes the proof.

\subsection{Proof of Lemma~\ref{lem:structure}}
\label{pf_lem:structure}
As illustrated by matrix $\textcircled{\scriptsize 2}$ and $\textcircled{\scriptsize 2}'$ in Fig.~\ref{fig_symid_sub}, we can separate $\mathbf{H}_{\mbox{\scriptsize sub}}$ into two matrices, one consists of the diagonal sub-blocks and sub-blocks that are not independent of the diagonal blocks and the other consists of the sub-blocks that are independent of the diagonal sub-blocks. It is sufficient to show that the first matrix is full rank almost surely. From S3, each diagonal block is full rank almost surely. Hence, as illustrated by matrix $\textcircled{\scriptsize 3}$ in Fig.~\ref{fig_symid_sub}, we can sequentially use row operation to make sub-blocks $\{\mathbf{B}_{(\!d\!-\!1\!)dkk'}\}$, $\{\mathbf{B}_{(\!d\!-\!2\!)(\!d\!-\!1\!)kk'}\}$,..., $\{\mathbf{B}_{12kk'}\}$ equal to $\mathbf{0}$ and make the matrix block upper-triangular. Noticing the association pattern S5, and the inter-block independence property S4, these operations preserves the block diagonal structure S2 and the full rankness of the diagonal sub-blocks. Then we can further adopt row operation to make the matrix block diagonal, e.g. matrix $\textcircled{\scriptsize 4}$ in Fig.~\ref{fig_symid_sub}. Since each sub-block is full rank almost surely, the entire matrix is full rank almost surely.

\subsection{Proof of Lemma~\ref{lem:suf1}}
\label{pf_lem_cor:sym}
\mysubnote{Give an example so that readers can ``visualize" the main idea of the proof.}
\subsubsection{Illustration of the Proof} We first illustrate the outline of
the proof via an example and give out the rigorous proof in the next subsection.

Consider a 3-pairs MIMO interference network with configuration $\chi=\{(2,2,1),(2,2,1),(4,2,2)\}$. The constraint allocation $\{c^{\mathrm{t}}_{kjpq},c^{\mathrm{r}}_{kjpq}\}$ is given by \eqref{eqn:callocation}.
\begin{eqnarray}
\left\{\!\!\!\!\begin{array}{*{5}{l@{,\,}}l}
c^{\mathrm{r}}_{1211}&c^{\mathrm{r}}_{2111}&c^{\mathrm{r}}_{3111}&\multicolumn{2}{l}{\!\!\!c^{\mathrm{r}}_{3221}=1}
\\c^{\mathrm{r}}_{1311}&c^{\mathrm{r}}_{1312}&c^{\mathrm{r}}_{2311}&c^{\mathrm{r}}_{2312}&c^{\mathrm{r}}_{3121}&c^{\mathrm{r}}_{3211}=0
\\c^{\mathrm{t}}_{1211}&c^{\mathrm{t}}_{2111}&c^{\mathrm{t}}_{3111}&\multicolumn{2}{l}{\!\!\!c^{\mathrm{t}}_{3221}=0}
\\c^{\mathrm{t}}_{1311}&c^{\mathrm{t}}_{1312}&c^{\mathrm{t}}_{2311}&c^{\mathrm{t}}_{2312}&c^{\mathrm{t}}_{3121}&c^{\mathrm{t}}_{3211}=1
\end{array}\right.\label{eqn:callocation}
\end{eqnarray}

From Thm.~\ref{thm:feasible_s}, to prove the lemma, we only need to show that
$\mathbf{H}_{\mathrm{all}}$ is full row rank almost surely.

\begin{figure} \centering
\includegraphics[scale=0.65]{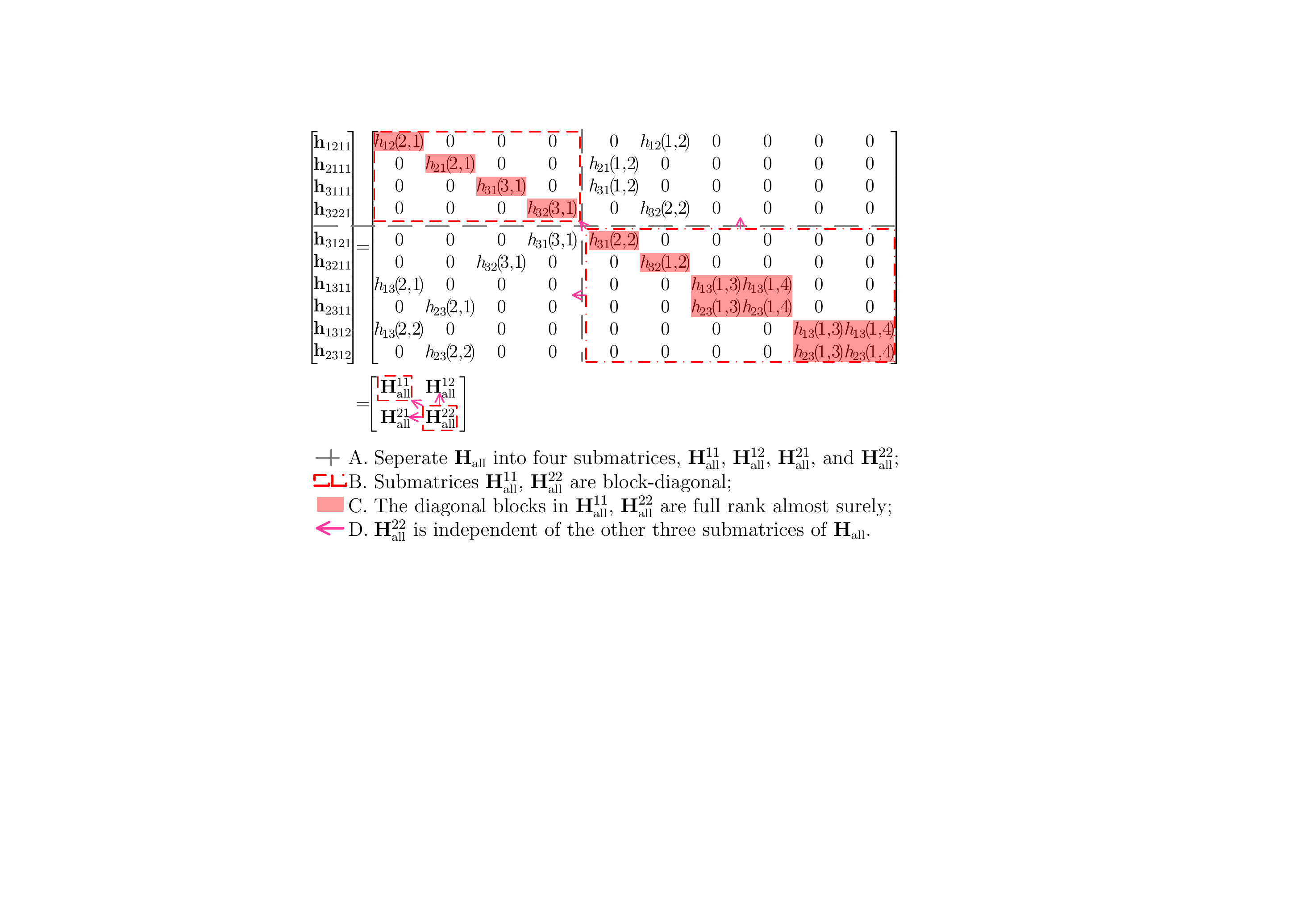}
\caption {Illustration of the matrix aggregated by vectors $\{\mathbf{h}_{kjpq}\}$ and its properties. Here $h_{kj}(p,q)$, $\tilde{u}_{k}(p,q)$, and $\tilde{v}_{j}(p,q)$ denote the element in the $p$-th row and $q$-th column of $\mathbf{H}_{kj}$, $\tilde{\mathbf{U}}_{k}$ and $\tilde{\mathbf{V}}_{j}$, respectively.}
\label{fig_egid}
\end{figure}

As illustrated by Fig.~\ref{fig_egid}, consider the matrix $\mathbf{H}_{\mathrm{all}}$ (we have rearranged the order of the rows for clearer illustration). In this network, we have 10 polynomials in \eqref{eqn:czero_poly2} and 10 variables in $\tilde{\mathbf{U}}_{k},\tilde{\mathbf{V}}_{j}$. Hence $\mathbf{H}_{\mathrm{all}}$ is a $10\times 10$ {matrix}. We need to prove that $\mathbf{H}_{\mathrm{all}}$ is nonsingular, i.e. $\det\left(\mathbf{H}_{\mathrm{all}}\right)\neq 0$ almost surely.

The major properties of $\mathbf{H}_{\mathrm{all}}$ that lead to its nonsingularity are labeled in Fig.~\ref{fig_egid}. We first carefully arrange the order of vectors $\{\mathbf{h}_{kjpq}\}$. In particular, index sequences $(k,j,p,q)$ that satisfy $c^{\mathrm{r}}_{kjpq}=1$ or $c^{\mathrm{t}}_{kjpq}=1$  are placed in the upper and lower part of $\mathbf{H}_{\mathrm{all}}$, respectively. We partition $\mathbf{H}_{\mathrm{all}}$ by rows according to  whether $c^{\mathrm{r}}_{kjpq}=1$ or $c^{\mathrm{t}}_{kjpq}=1$, and by columns according to whether the column is occupied by $\{\mathbf{H}^{\mathrm{U}}_{kj}\}$ or $\{\mathbf{H}^{\mathrm{V}}_{kj}\}$. Then as illustrated by Label A in Fig.~\ref{fig_egid}, $\mathbf{H}_{\mathrm{all}}$ is partitioned into four submatrices. $\mathbf{H}^{11}_{\mathrm{all}}$ and $\mathbf{H}^{22}_{\mathrm{all}}$ are block-diagonal as we have reordered the rows in $\mathbf{H}_{\mathrm{all}}$. As highlighted by Label C, all the diagonal blocks of $\mathbf{H}^{11}_{\mathrm{all}}$, $\mathbf{H}^{22}_{\mathrm{all}}$ are full rank almost surely. Thus we have $\mathbf{H}^{11}_{\mathrm{all}}$, $\mathbf{H}^{22}_{\mathrm{all}}$ are nonsingular almost surely, i.e. $\det{\mathbf{H}^{11}_{\mathrm{all}}}\neq 0,\; \det{\mathbf{H}^{22}_{\mathrm{all}}}\neq 0$, almost surely.

From condition \eqref{eqn:cv3}, if an element $h_{kj}(p,q)$ appears in $\mathbf{H}^{22}_{\mathrm{all}}$, it will not appear in other sub-matrices of $\mathbf{H}_{\mathrm{all}}$. Hence, as illustrated by Label D in Fig.~\ref{fig_egid}, $\mathbf{H}^{22}_{\mathrm{all}}$ is independent of the other three submatrices. Then from the Leibniz formula, we have that $\det(\mathbf{H}_{\mathrm{all}})\neq 0$ holds almost surely.

\mysubnote{Formal Proof.}
\subsubsection{Extending to General Configurations}
We will show that properties of $\mathbf{H}_{\mathrm{all}}$ illustrated by Labels A{--}D in Fig.~\ref{fig_egid} hold for generic configurations.

Denote $\{c^{t*}_{kjpq}, c^{r*}_{kjpq}\}$ as a set binary variables that satisfy \eqref{eqn:cv1}{--}\eqref{eqn:cv3}. {Without loss of generality}, suppose condition \eqref{eqn:cv3} holds for $\{c^{t*}_{kjpq}\}$.

 Reorder the rows of $\mathbf{H}_{\mathrm{all}}$ such that the row vectors which satisfy $\{\mathbf{h}_{kjpq}\!:\!c^{r*}_{kjpq}\!=\!1\}$ appear in the upper part of the matrix. To show that $\mathbf{H}_{\mathrm{all}}$ is full row rank, we need to show there exists a sub-matrix $\mathbf{H}^\dag_{\mathrm{all}}\in\mathbb{C}^{(\sum_{k=1}^K\sum_{j=1,\neq k}^Kd_kd_j)\!\times\! (\sum_{k=1}^K\sum_{j=1,\neq k}^Kd_kd_j)}$ of $\mathbf{H}_{\mathrm{all}}$, whose determinant is non-zero almost surely. Construct $\mathbf{H}^\dag_{\mathrm{all}}$ by removing the columns which contain the coefficients of $\tilde{u}_{k}(p,q)$, $\tilde{v}_{j}(p^\dag,q^\dag)$, where $k$, $j$, $q$, $q^\dag$, $p$, and $p^\dag$ satisfy:

\begin{eqnarray}1+c^{r*}_{kp}&\le p\le& N_k-d_k,\label{eqn:extrar}\\
1+c^{t*}_{jp^\dag}&\le p^\dag\le& M_j-d_j\label{eqn:extrat},
\end{eqnarray}
where $c^{r*}_{kp}=\sum_{j=1\atop\neq k}^K \sum_{q=1}^{d_j} c^{r*}_{kjpq}$, $c^{t*}_{jq}=
\sum_{k=1\atop\neq j}^K \sum_{p=1}^{d_k}c^{t*}_{kjpq}$.

In the following analysis, we will partition $\mathbf{H}_{\mathrm{all}}$ in the same way as that in the example above and show that the major properties labeled in Fig.~\ref{fig_egid} still {hold}.
\begin{itemize}
\item{\bf $\mathbf{H}^{11}_{\mathrm{all}}$ is full rank almost surely:} Consider the $1\sim c^{r*}_{11}$ columns of $\mathbf{H}^\dag_{\mathrm{all}}$. From \eqref{eqn:extrar}, these columns are also the first $c^{r*}_{11}$ (note that from \eqref{eqn:cv2r}, $c^{r*}_{11}\le N_1-d_1$) columns of $\mathbf{H}_{\mathrm{all}}$. Hence from the form of vectors $\{\mathbf{h}_{kjpq}\}$, we have that elements in these columns are non-vanishing if and only if index $k=1$, $j=1$. In $\mathbf{H}^{11}_{\mathrm{all}}$, only the first $c^{r*}_{11}$ rows are non-vanishing. Repeat the same analysis for other columns and we can show that $\mathbf{H}^{11}_{\mathrm{all}}$ is a $\sum_{k=1}^K \sum_{p=1}^{d_k} c^{r*}_{kp}\times \sum_{k=1}^K\sum_{p=1}^{d_k} c^{r*}_{kp}$ block diagonal matrix, with diagonal block sizes $c^{r*}_{11},...,c^{r*}_{1d_1},...,c^{r*}_{Kd_K}$, respectively. From Fig.~\ref{fig_Hall}, \eqref{eqn:hu}, and \eqref{eqn:hv}, we have that the elements in a same diagonal blocks are independent of each other, hence $\det{\mathbf{H}^{11}_{\mathrm{all}}}\neq 0$ almost surely.
\item{\bf $\mathbf{H}^{22}_{\mathrm{all}}$ is full rank almost surely:} Using the analysis similar to above, we can show that $\mathbf{H}^{22}_{\mathrm{all}}$ is also block-diagonal and $\det{\mathbf{H}^{22}_{\mathrm{all}}}\neq 0$ almost surely.
\item{\bf $\mathbf{H}^{22}_{\mathrm{all}}$ is independent of $\mathbf{H}^{11}_{\mathrm{all}}$, $\mathbf{H}^{12}_{\mathrm{all}}$, and $\mathbf{H}^{21}_{\mathrm{all}}$:} From Fig.~\ref{fig_Hall}, \eqref{eqn:hu}, and \eqref{eqn:hv}, an element $h_{kj}(p,d_j+s)$, $p\in\{1,...,d_k\}$, $s\in\{1,...M_j-d_j\}$ only appears in vectors $\mathbf{h}_{kjp1},...,\mathbf{h}_{kjpd_j}$. Hence condition \eqref{eqn:cv3} assure that if $h_{kj}(p,d_j+s)$ appears in $\mathbf{H}^{22}_{\mathrm{all}}$, it appears in neither of the other three sub-matrices. This proves that $\mathbf{H}^{22}_{\mathrm{all}}$ is independent of $\mathbf{H}^{11}_{\mathrm{all}}$, $\mathbf{H}^{12}_{\mathrm{all}}$, and $\mathbf{H}^{21}_{\mathrm{all}}$.
\end{itemize}

The three facts above show that $\det\left(\mathbf{H}^\dag_{\mathrm{all}}\right)\neq 0$ almost surely. This completes the proof.

\subsection{Proof of Lemma~\ref{lem:nes1}}
\label{pf_lem:nes1}

\mysubnote{Prove the ``only if" statement.}

We first prove the ``only if"
side. Denote $\{c^{t*}_{kjpq},c^{r*}_{kjpq}\}$ as a set of binary variables that satisfy \eqref{eqn:cv1}{--}\eqref{eqn:cv2t}. Then we have:
\begin{eqnarray}&&\!\!\!\!\!\!\!\!\!\!\!\!\nonumber\sum_{j:(\cdot,j)\in \mathcal{J}_{\mathrm{sub}}}\!\!\!\!\!\!d_j(M_j-d_j)
+\!\!\!\!\!\!\sum_{k:(k,\cdot)\in \mathcal{J}_{\mathrm{sub}}}\!\!\!\!\!\!d_k(N_k-d_k)
\\&& \!\!\!\!\!\!\qquad\quad \ge
\!\!\!\!\sum_{j:(\cdot,j)\atop\in \mathcal{J}_{\mathrm{sub}}}\!\sum_{k:(k,\cdot)\atop\in \mathcal{J}_{\mathrm{sub}}}\sum_{q=1}^{d_j}\sum_{p=1}^{d_k}(c^{\mathrm{t}}_{kjpq}+c^{\mathrm{r}}_{kjpq})\label{eqn:c2c}
\\&&\!\!\!\!\!\! \qquad\quad=\!\!\!\!\sum_{(k,j)\in \mathcal{J}_{\mathrm{sub}}}\!\!\!\!d_k d_j
\label{eqn:ifside}
\end{eqnarray}
$\forall \mathcal{J}_{\mathrm{sub}}$, where \eqref{eqn:c2c} is true due to \eqref{eqn:cv2r}, \eqref{eqn:cv2t}, and \eqref{eqn:ifside} is given by \eqref{eqn:cv1}. This completes the ``only if" side of the proof.

\mysubnote{Prove the ``if" statement.}

Then we turn to the ``if" side. We adopt a constructive approach. In following algorithm\footnote{ Note that this algorithm may not be the only way to construct $\{c^{t*}_{kjpq}, c^{r*}_{kjpq}\}$.}, we will propose a method to construct the binary variables $\{c^{t*}_{kjpq}, c^{r*}_{kjpq}\}$ and show that when conditions \eqref{eqn:f3} is true, the binary variables $\{c^{t*}_{kjpq}, c^{r*}_{kjpq}\}$ constructed by Alg.~\ref{alg:c} satisfy \eqref{eqn:cv1}{--}\eqref{eqn:cv2t}.

\mysubnote{This algorithm is the core component of the constructive proof.}
{
We first define two notions which will be used in the algorithm. To indicate how far the constraint assignment policy $\{c^{\mathrm{r}}_{kjpq}, c^{\mathrm{t}}_{kjpq}\}$  is from satisfying constraints \eqref{eqn:cv2r}, \eqref{eqn:cv2t}, we define:
\begin{Def}[Constraint Pressure]
\begin{eqnarray}
 P^{\mathrm{t}}_{jq}\!\triangleq\!M_j\!-\!d_j\!-\!\!\sum_{k=1\atop \neq j}^K\!\sum_{p=1}^{d_k} \!c^{\mathrm{t}}_{kjpq},\,P^{\mathrm{r}}_{kp}\!\triangleq\!N_k\!-\!d_k\!-\!\!\sum_{j=1\atop \neq k}^K\!\sum_{q=1}^{d_j} \!c^{\mathrm{r}}_{kjpq}.\label{eqn:pressure}\end{eqnarray}~\hfill~\IEEEQED
\end{Def}

To update the constraint assignment policy, we introduce the following data structure.
\begin{Def}[Pressure Transfer Tree (PTT)] As illustrated by Fig.~\ref{fig_tree}A, define
a weighted tree data structure with the following properties:
\begin{itemize}
\item[-]The weight of the nodes in this tree is given by the constraint pressure, i.e. $\{P^{\mathrm{t}}_{jq}, P^{\mathrm{r}}_{kp}\}$.
\item[-]The constraint pressure of a parent node and its child nodes have different superscript, i.e. $t$ or $r$.
\item[-]The link strength between two nodes, e.g. $P^{\mathrm{t}}_{jq}$ and $P^{\mathrm{r}}_{kp}$ is given by $c^{\mathrm{t}}_{kjpq}$, if $P^{\mathrm{t}}_{jq}$ is the parent node, or $c^{\mathrm{r}}_{kjpq}$, if $P^{\mathrm{r}}_{kp}$ is the parent node.~\hfill~\IEEEQED
\end{itemize}
\end{Def}
}
\begin{Alg}[Construct $\{c^{t*}_{kjpq}, c^{r*}_{kjpq}\}$]\label{alg:c}
\begin{itemize}
\item[]
\item{\bf Initialize the constraint allocation:} Randomly {generate}
a \emph{constraint allocation policy}, i.e. $\{c^{\mathrm{t}}_{kjpq},c^{\mathrm{r}}_{kjpq}\}$
such that: $c^{\mathrm{t}}_{kjpq},c^{\mathrm{r}}_{kjpq}\in\{0,1\}$, $
c^{\mathrm{t}}_{kjpq}+c^{\mathrm{r}}_{kjpq}=1$, $k,j \in\{1,2,...,K\}$, $p\in\{1,...,d_k\}$, $q\in\{1,...,d_j\}$. Calculate $\{P^{\mathrm{t}}_{jq},P^{\mathrm{r}}_{kp}\}$ according to \eqref{eqn:pressure}.
\item{\bf Update the constraint allocation policy:} While there exist ``overloaded streams", i.e. $P^{\mathrm{t}}_{jq}<0$ or
$P^{\mathrm{r}}_{kp}<0$, do the following to update $\{c^{\mathrm{t}}_{kjpq},c^{\mathrm{r}}_{kjpq}\}$:
\begin{itemize}\item{\bf A. Initialization:}
Select a negative pressure, e.g. $P^{\mathrm{t}}_{jq}<0$. Set $P^{\mathrm{t}}_{jq}$
to be the root node of a PTT.
\item{\bf B. Add Leaf nodes to the
pressure transfer tree:}

For every leaf nodes (i.e. nodes without child nodes) e.g. $P^{\mathrm{t}}_{jq}$, with depths equal to the height
of the tree (i.e. the nodes at the bottom in Fig.~\ref{fig_tree}):
\begin{itemize}
\item[] For every $k\in\{1,2,...,K\}$, $p\in\{1,...,d_k\}$: If $c^{\mathrm{r}}_{kjpq}=1$,
add $P^{\mathrm{r}}_{kp}$ as a child node of $P^{\mathrm{t}}_{jq}$.
\end{itemize}
\item{\bf C. Transfer pressure from root to leaf nodes:} For every leaf node just
added to the tree in Step B with positive pressure, i.e. $P^{\mathrm{t}}_{jq} (\mbox{ or } P^{\mathrm{r}}_{kp})>0$, transfer
pressure from root to these leafs by updating $\{c^{\mathrm{t}}_{kjpq},c^{\mathrm{r}}_{kjpq}\}$. For instance, as
illustrated in Fig.~\ref{fig_tree}B,
$P^{\mathrm{t}}_{j_1q_1}\xrightarrow{c^{\mathrm{t}}_{k_1j_1p_1q_1}\!=\!1}P^{\mathrm{r}}_{k_1p_1}\xrightarrow{c^{\mathrm{r}}_{k_1j_2p_1q_2}\!=\!1}
P^{\mathrm{t}}_{j_2q_2}$
is a root-to-leaf branch of the tree (red lines). Transfer pressure
from $P^{\mathrm{t}}_{j_1q_1}$ to $P^{\mathrm{t}}_{j_2q_2}$ by setting: $c^{\mathrm{t}}_{k_1j_1p_1q_1}=
0$, $c^{\mathrm{r}}_{k_1p_1j_1q_1}=
1$, $c^{\mathrm{r}}_{k_1j_2p_1q_2}= 0$,
$c^{\mathrm{t}}_{k_1j_2p_1q_2}= 1$. Hence we have
$P^{\mathrm{t}}_{j_1q_1}$ is increased by $1$ and
$P^{\mathrm{t}}_{j_2q_2}$ is reduced by $1$. This
operation can also be done for the green line in
Fig.~\ref{fig_tree}B.
\item{\bf D. Remove the ``depleted" links and ``neutralized" roots:}
\begin{itemize}
\item If the strength of a link, i.e. $c^{\mathrm{r}}_{kjpq}$ or $c^{\mathrm{t}}_{kjpq}$, becomes 0 after Step C: Separate the
subtree rooted from the child node of this link from the original
tree.
\item If the root of a pressure transfer tree (including the
subtrees just separated from the original tree) is nonnegative,
remove the root and hence the subtrees rooted from each child node
of the root become new trees. Repeat this process until all roots
are negative. For each newly generated pressure transfer tree,
repeat Steps B{--}D (Please refer to Fig.~\ref{fig_tree}C for an
example).
\end{itemize}
\item{\bf E. Exit Conditions:} Repeat Steps A{--}D until one of the following cases appears.

{\bf Case 1:} All
trees become empty.

{\bf Case 2:} No new leaf
node can added for any of the non-empty trees in Step B.

Set $\{c^{t*}_{kjpq}, c^{r*}_{kjpq}\}$ to be the current value of $\{c^{\mathrm{t}}_{kjpq}, c^{\mathrm{r}}_{kjpq}\}$ Exit the algorithm.~\hfill~\IEEEQED
\end{itemize}
\end{itemize}
\end{Alg}
\begin{figure*}[t] \centering
\includegraphics[scale=0.55]{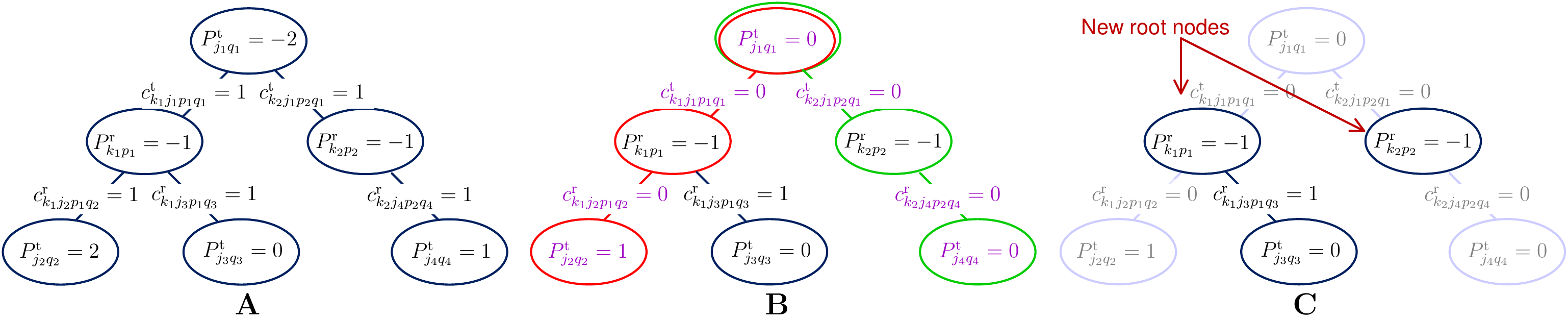}
\caption {Illustrative example of the ``pressure
transfer tree" and the corresponding operations in
Alg.~\ref{alg:c}. A) A tree generated in Step A and B;
B) Pressure transfer in Step C; C) Removal of depleted links and
neutralized roots in Step D.} \label{fig_tree}
\end{figure*}

Note that if Alg.~\ref{alg:c} exits with Case 1, we have that $M_j-d_j-\sum_{k=1\atop \neq j}^K\sum_{p=1}^{d_k} c^{\mathrm{t}}_{kjpq}\ge 0$, and
$N_k-d_k-\sum_{j=1\atop \neq k}^K\sum_{q=1}^{d_j} c^{\mathrm{r}}_{kjpq}\ge 0$, $\forall k,j,p,q$, which lead to \eqref{eqn:f3}. Hence to prove the ``if" side, we only need to prove the following proposition:
\begin{Prop2}[Exit State of Alg.~\ref{alg:c}]\hspace{-40.9mm}{\color{white}\scriptsize$\blacklozenge$}\hspace{38.7mm} When \eqref{eqn:f3} is true, Alg.~\ref{alg:c} always exits with Case 1.\label{prop:exit}
\end{Prop2}

We prove Prop\ref{prop:exit}\hspace{-6.6mm}{\color{white}\scriptsize$\blacklozenge$}\hspace{4.4mm} by contradiction. \eqref{eqn:f3} is equivalent to the following inequalities:
\vspace{-5pt}
\begin{eqnarray}
\!\!\!\!\sum_{(j,q):(\cdot,j,\cdot,q)\atop \in\mathcal{L}_{\mathrm{sub}}}\!\!\!\!\!\!(M_j\!-\!d_j)
\!+\!\!\!\!\!\!\sum_{(k,p):(k,\cdot,p,\cdot)\atop \in\mathcal{L}_{\mathrm{sub}}}\!\!\!\!\!\!(N_k\!-\!d_k)\!\ge\! |\mathcal{L}_{\mathrm{sub}}|,\,\forall \mathcal{L}_{\mathrm{sub}}\!\subseteq\! \mathcal{L}\label{eqn:f3_new}
\end{eqnarray}
\vspace{-10pt}

$\!\!\!\!\!\!$where $\mathcal{L}=\{(k,j,p,q):\; k,j\in\{1,...,K\},k\neq j,p\in\{1,...,d_k\},q\in\{1,...,d_j\}\}$, $(\cdot,j,\cdot,q)$ (or $(k,\cdot,p,\cdot)$) $\ in\mathcal{L}_{\mathrm{sub}}$ denotes that there exists $k,p$ (or $j,q$) such that $(k,j,p,q)\in\mathcal{L}_{\mathrm{sub}}$.
If Alg.~\ref{alg:c} exits with Case 2, from the
exit condition, there must exist a non-empty pressure
transfer tree such that:
\begin{itemize}
\item Root node has negative pressure.
\item All other nodes are non-positive. This is because positive nodes are either
``neutralized" by the root in Step C or separated from the
tree in Step D.
\item No other nodes can be added to the tree, which implies $c^{\mathrm{r}}_{kjpq}=0$ and $c^{\mathrm{t}}_{k^\dag j^\dag p^\dag q^\dag}=0$ for any
$P^{\mathrm{t}}_{jq}$, $P^{\mathrm{r}}_{k^\dag p^\dag}$  in the tree and $P^{\mathrm{r}}_{kp}$, $P^{\mathrm{t}}_{j^\dag q^\dag}$ not in the tree.
\end{itemize}
Hence, set $\mathcal{L}_{\mathrm{sub}}$ in \eqref{eqn:f3_new} to be
$\{(k,j,p,q):$ both $P^{\mathrm{r}}_{kp}$  and  $P^{\mathrm{t}}_{jq}$ are in the tree$\}$ that are in the remaining pressure
transfer tree, we have:
\begin{eqnarray}
\nonumber &&\!\!\!\!\!\!\!\!\!\!\!\!\!\!\!\!\!\!\!\! \sum_{(j,q):(\cdot,j,\cdot,q)\atop\in \mathcal{L}_{\mathrm{sub}}}\!\!\!\!(M_j\!-\!d_j)
\!+\!\!\!\!\!\sum_{(k,p):(k,\cdot,p,\cdot)\atop\in \mathcal{L}_{\mathrm{sub}}}\!\!\!\!(N_k\!-\!d_k)\!-\! |\mathcal{L}_{\mathrm{sub}}|
\\ \nonumber &&\!\!\!\!\!\!\!\!\!\!\!\!\!\!\!\!\!\!\!\! \qquad\!=\!\!\!\!\!\!\!\!\sum_{(j,q):(\cdot,j,\cdot,q)\atop\in \mathcal{L}_{\mathrm{sub}}}\!\!\!\!\!\!\!\!(M_j\!-\!d_j\!-\!\!\!\!\!\!\!\!\sum_{(k,p):(k,\cdot,p,\cdot)\atop\in \mathcal{L}_{\mathrm{sub}}}\!\!\!\!\!\!\!\!c^{\mathrm{t}}_{kjpq})\!+\!
\!\!\!\!\!\!\!\sum_{(k,p):(k,\cdot,p,\cdot)\atop\in \mathcal{L}_{\mathrm{sub}}}\!\!\!\!\!\!\!\!(N_k\!-\!d_k\!-\!\!\!\!\!\!\!\!\sum_{(j,q):(\cdot,j,\cdot,q)\atop\in \mathcal{L}_{\mathrm{sub}}}\!\!\!\!\!\!\!\!c^{\mathrm{r}}_{kjpq})
\\ \nonumber &&\!\!\!\!\!\!\!\!\!\!\!\!\!\!\!\!\!\!\!\!\qquad\!=\!
\!\!\!\!\!\!\!\!\!\sum_{(j,q): (k,j,p,q)\atop\in \mathcal{L}_{\mathrm{sub}}}\!\!\!\!\!\!\!\!(M_j\!-\!d_j\!-\!\!\!\sum_{k=1}^K\sum_{p=1}^{d_k} \!c^{\mathrm{t}}_{kjpq})\!+\!
\!\!\!\!\!\!\!\!\!\sum_{(k,p):(k,\cdot,p,\cdot)\atop \in \mathcal{L}_{\mathrm{sub}}}\!\!\!\!\!\!\!\!(N_k\!-\!d_k\!-\!\!\!\sum_{j=1}^K\sum_{q=1}^{d_j}\! c^{\mathrm{r}}_{kjpq})
\\&&\!\!\!\!\!\!\!\!\!\!\!\!\!\!\!\!\!\!\!\! \qquad\!=\!\!\!\!\!\sum_{(j,q):(\cdot,j,\cdot,q)\atop\in \mathcal{L}_{\mathrm{sub}}}\!\!\!\!P^{\mathrm{t}}_{jq}
\!+\!\!\!\!\!\sum_{(k,p):(k,\cdot,p,\cdot)\atop\in \mathcal{L}_{\mathrm{sub}}}\!\!\!\!P^{\mathrm{r}}_{kp}\!<\!0 \label{eqn:otherside}
\end{eqnarray}
which contradicts with \eqref{eqn:f3_new}.

\subsection{Construct $\{c^{t*}_{kjpq}, c^{r*}_{kjpq}\}$ to Prove Corollary~\ref{cor:div} }
\label{pf_constructc}
{Without loss of generality}, assume $d|N_k$, $\forall k$. Modify Alg.~\ref{alg:c} in Appendix\ref{pf_lem:nes1}\hspace{-4mm}{\color{white}\scriptsize$\blacklozenge$}\hspace{2.1mm} to construct $\{c^{t*}_{kjpq}, c^{r*}_{kjpq}\}$.
\begin{Alg}[Variation of Algorithm~\ref{alg:c}]
Adopt Alg.~\ref{alg:c} with the following modifications:
\begin{itemize}
\item {In the Initialization Step:} set
$c^{\mathrm{t}}_{kjp1}=...=c^{\mathrm{t}}_{kjpd_j}$, $c^{\mathrm{r}}_{kjp1}=...=c^{\mathrm{r}}_{kjpd_j}$.
\item {In Step C:} All pressure transfer operations must be symmetric w.r.t to index $q$, i.e. the operations on $c^{\mathrm{t}}_{kj1q}$ ($c^{\mathrm{r}}_{kj1q}$), $...$, $c^{\mathrm{t}}_{kjd_kq}$ ($c^{\mathrm{r}}_{kjd_kq}$) must be the same.~\hfill~\IEEEQED
\end{itemize}\label{alg:c2}
\end{Alg}

Suppose Prop\ref{prop:exit}\hspace{-6.6mm}{\color{white}\scriptsize$\blacklozenge$}\hspace{4.4mm} still holds for Alg.~\ref{alg:c2}. Then the output of Alg.~\ref{alg:c2} $\{c^{t*}_{kjpq}, c^{r*}_{kjpq}\}$ satisfies \eqref{eqn:cv1}{--}\eqref{eqn:cv3}. Therefore, we focus on proving Prop.\ref{prop:exit}\hspace{-6.6mm}{\color{white}\scriptsize$\blacklozenge$}\hspace{4.4mm} .

From \eqref{eqn:pressure}, after initialization, we have that
\begin{eqnarray}P^{\mathrm{t}}_{j1}=...=P^{\mathrm{t}}_{jd_j},\; d|P^{\mathrm{t}}_{jq},\; \forall n,q.
\label{eqn:divP}\end{eqnarray}

Since all $c^{\mathrm{t}}_{kjpq}$, $c^{\mathrm{r}}_{kjpq}$, $P^{\mathrm{t}}_{jq}$, and $P^{\mathrm{r}}_{kp}$ are symmetric w.r.t. to index $q$, after we perform Steps A, B of the constraint updating process, the pressure transfer trees are also symmetric w.r.t. to index $q$. Further note that $d|P^{\mathrm{t}}_{jq}$, in Step C, the symmetric operation is always feasible. As a result, we can follow the analysis used in Appendix\ref{pf_lem:nes1}\hspace{-4mm}{\color{white}\scriptsize$\blacklozenge$}\hspace{2.1mm}, and prove Prop\ref{prop:exit}\hspace{-6.5mm}{\color{white}\scriptsize$\blacklozenge$}\hspace{4.3mm}.

\vspace{-10pt}
\begin{IEEEbiography}
[{\includegraphics[width=1in,height=1.25in,clip,keepaspectratio]{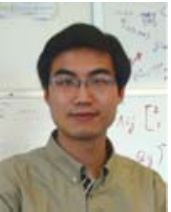}}]{Liangzhong Ruan} (S'10)
received B.Eng degree in ECE from Tsinghua University, Beijing in 2007.
He is currently a Ph.D. student at the Department of ECE,
Hong Kong University of Science and Technology (HKUST). Since Feb. 2012, he has been a visiting graduate student at the Laboratory for Information \& Decision Systems (LIDS), Massachusetts Institute of Technology (MIT). His research interests include interference management and
 cross-layer optimization.
\end{IEEEbiography}
\vspace{-8pt}
\begin{IEEEbiography}[{\includegraphics[width=1in,height=1.25in,clip,keepaspectratio]
{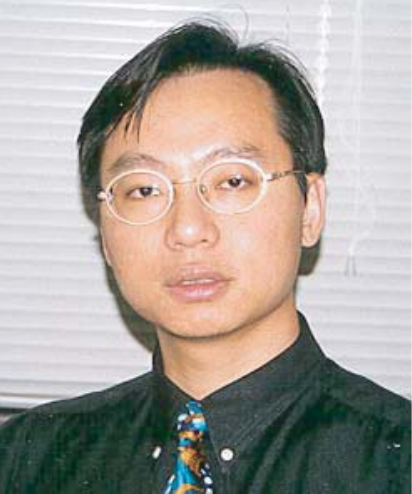}}]{Vincent Lau} (F'11) obtained B.Eng (Distinction 1st Hons) from the University of Hong Kong (1989-1992)
and Ph.D. from Cambridge University (1995-1997). Dr. Lau joined the department of ECE, HKUST as Associate Professor in 2004 and was promoted to Professor in 2010. His research interests include the
delay-sensitive cross-layer optimization for wireless systems, interference management, cooperative communications, as well as stochastic approximation and Markov Decision Process. He has been the technology advisor and consultant for a number of companies such as ZTE, Huawei, and ASTRI. He is the founder and director of Huawei-HKUST Innovation Lab.
\end{IEEEbiography}
\vspace{-8pt}
\begin{IEEEbiography}[{\includegraphics[width=1in,height=1.25in,clip,keepaspectratio]
{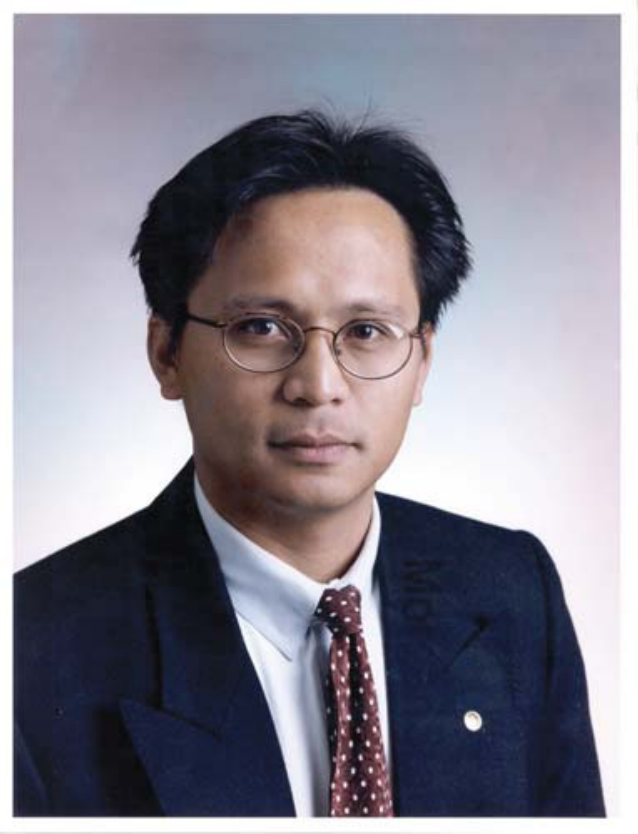}}]{Moe Win} (F'04)
 received the B.S. degree from Texas A\&M University, College Station, in 1987 and the M.S. degree from the University of Southern California (USC), Los Angeles, in 1989, both in Electrical Engineering. As a Presidential Fellow at USC, he received both an M.S. degree in Applied Mathematics and the Ph.D. degree in Electrical Engineering in 1998.

Dr. Win is a Professor at the LIDS, MIT. His main research interests are the application of mathematical and statistical theories to communication, detection, and estimation problems. He is a recipient
of the IEEE Kiyo Tomiyasu Award and the IEEE Eric E. Sumner Award.
\end{IEEEbiography}

\end{document}